\documentclass[prb,amsmath,amssymb]{revtex4}

\usepackage{amsmath}
\usepackage{amssymb}
\usepackage{graphics}
\usepackage{pstricks}
\usepackage{epic,eepic}
\usepackage{fancyhdr}
\usepackage{verbatim}
\usepackage{latexsym}
\usepackage{enumerate}
\usepackage{mathrsfs}

\newcommand{\bfalpha}{{\boldsymbol{\alpha}}}
\newcommand{\calA}{\mathcal{A}}
\newcommand{\calH}{\mathcal{H}}
\newcommand{\calS}{\mathcal{S}}
\newcommand{\dd}{\mathrm{d}}
\newcommand{\ee}{\mathrm{e}}
\newcommand{\UGL}{\Omega}
\newcommand{\barsigma}{\bar\sigma}
\newcommand{\tildephi}{\tilde\varphi}
\newcommand{\WO}{\overline\Omega}
\newcommand{\baromega}{\overline{\omega}}

\newcommand{\defeq}{:=}

\newcommand{\Heff}{H_{\mathrm{eff}}}

\newtheorem{prop}{Proposition}
\newtheorem{thm}[prop]{Theorem}
\newtheorem{cor}[prop]{Corollary}
\newtheorem{lem}[prop]{Lemma}
\newtheorem{dfn}[prop]{Definition}
\newenvironment{proof}{\noindent{\em
Proof.\/}}{\hfill$\Box$\par\vspace{.2cm}}

\renewcommand{\|}{
\setlength{\unitlength}{3pt}
\psset{unit=3pt}
\psset{runit=2pt}
\psset{linewidth=0.2}
\begin{pspicture}(0,.5)(2.5,4)
\psline(1,0)(1,3)
\end{pspicture}}

\newcommand{\Y}{
\setlength{\unitlength}{3pt}
\psset{unit=3pt}
\psset{runit=2pt}
\psset{linewidth=0.2}
\begin{pspicture}(-1,.5)(3,4)
\psline(1,0)(1,2)
\psline(1,2)(0,3)
\psline(1,2)(2,3)
\end{pspicture}}

\newcommand\deuxun{
\setlength{\unitlength}{3pt}
\psset{unit=3pt}
\psset{runit=2pt}
\psset{linewidth=0.2}
\begin{pspicture}(0,.5)(5,5)
\psline(3,0)(3,2)
\psline(3,2)(1,4)
\psline(3,2)(4,3)
\psline(2,3)(3,4)
\end{pspicture}}

\newcommand\deuxdeux{
\setlength{\unitlength}{3pt}
\psset{unit=3pt}
\psset{runit=2pt}
\psset{linewidth=0.2}
\begin{pspicture}(1,.5)(6,5)
\psline(3,0)(3,2)
\psline(3,2)(5,4)
\psline(3,2)(2,3)
\psline(4,3)(3,4)
\end{pspicture}}

\newcommand\troisun{
\setlength{\unitlength}{3pt}
\psset{unit=3pt}
\psset{runit=2pt}
\psset{linewidth=0.2}
\begin{pspicture}(-1,.5)(5,6)
\psline(3,0)(3,2)
\psline(3,2)(0,5)
\psline(3,2)(4,3)
\psline(2,3)(3,4)
\psline(1,4)(2,5)
\end{pspicture}}

\newcommand\troisdeux{
\setlength{\unitlength}{3pt}
\psset{unit=3pt}
\psset{runit=2pt}
\psset{linewidth=0.2}
\begin{pspicture}(0,.5)(5,6)
\psline(3,0)(3,2)
\psline(3,2)(1,4)
\psline(3,2)(4,3)
\psline(2,3)(4,5)
\psline(3,4)(2,5)
\end{pspicture}}

\newcommand\troistrois{
\setlength{\unitlength}{3pt}
\psset{unit=3pt}
\psset{runit=2pt}
\psset{linewidth=0.2}
\begin{pspicture}(-0.5,.5)(6.5,6)
\psline(3,0)(3,2)
\psline(3,2)(0.5,4.5)
\psline(1.5,3.5)(2.5,4.5)
\psline(3,2)(5.5,4.5)
\psline(4.5,3.5)(3.5,4.5)
\end{pspicture}}

\newcommand\troisquatre{
\setlength{\unitlength}{3pt}
\psset{unit=3pt}
\psset{runit=2pt}
\psset{linewidth=0.2}
\begin{pspicture}(1,.5)(6,6)
\psline(3,0)(3,2)
\psline(3,2)(5,4)
\psline(3,2)(2,3)
\psline(4,3)(2,5)
\psline(3,4)(4,5)
\end{pspicture}}

\newcommand\troiscinq{
\setlength{\unitlength}{3pt}
\psset{unit=3pt}
\psset{runit=2pt}
\psset{linewidth=0.2}
\begin{pspicture}(1,.5)(7,6)
\psline(3,0)(3,2)
\psline(3,2)(6,5)
\psline(3,2)(2,3)
\psline(4,3)(3,4)
\psline(5,4)(4,5)
\end{pspicture}}

\begin{document}

\title{Tree expansion in time-dependent perturbation theory}
\author{Christian Brouder and {\^A}ngela Mestre}
\affiliation{Institut de Min\'eralogie et de Physique des Milieux Condens\'es,
CNRS UMR 7590, Universit\'es Paris 6 et 7, IPGP, 140 rue de Lourmel,
75015 Paris, France.}
\author{Fr\'ed\'eric Patras}
\affiliation{Laboratoire J.-A. Dieudonn\'e, CNRS UMR 6621,
Universit\'e de Nice,
Parc Valrose, 06108 Nice Cedex 02, France.
}
\date{\today}
%
%
\begin{abstract}
The computational complexity of time-dependent perturbation theory is 
well-known to be largely combinatorial whatever the chosen expansion 
method and family of parameters (combinatorial sequences, Goldstone and 
other Feynman-type diagrams...). We show that a very efficient perturbative 
expansion, both for theoretical and numerical purposes, can be obtained 
through an original parametrization by trees and generalized iterated 
integrals. We emphasize above all the simplicity and naturality of the new 
approach that links perturbation theory with classical and recent results in 
enumerative and algebraic combinatorics. These tools are applied to the
adiabatic approximation and the effective Hamiltonian. We prove 
perturbatively and non-perturbatively
the convergence of Morita's generalization of the Gell-Mann and Low
wavefunction. We show that summing all the terms associated to the same
tree leads to an utter simplification where the sum is simpler than
any of its terms. Finally, we recover the time-independent equation for 
the wave operator and we give an explicit non-recursive expression for the term
corresponding to an arbitrary tree. 
\end{abstract}
\maketitle

\section{Introduction}
Effective Hamiltonians provide a way to determine the
low-energy eigenvalues of a (possibly infinite dimensional)
Hamiltonian by diagonalizing a matrix defined in a 
subspace of small dimension, called the
\emph{model space} and hereafter denoted by $M$.
Because of this appealing feature, effective Hamiltonians are used in 
nuclear, atomic, molecular, chemical and solid-state 
physics~\cite{KuoOsnes}. 

These theories are
plagued with a tremendous combinatorial complexity because
of the presence of folded diagrams (to avoid singularities of the 
adiabatic limit), partial resummations, 
subtle ``linkedness'' properties and the 
exponential growth of the number of graphs with the
order of perturbation.
This complexity has two consequences: on the one hand, 
few results are proved in the mathematical sense of the word,
on the other hand, it is difficult to see what is the underlying 
structure of the perturbative expansion that could lead to 
useful resummations and non-perturbative approximations.

To avoid these pitfalls, we take a bird's-eye view of the problem 
and consider a general time-dependent Hamiltonian
$H(t)$. This way, we disentangle the problem from the various particular 
forms that can be given to the Hamiltonian and which have lead in the past 
to various perturbative expansions. 
To be precise, take the example of fermions in molecular systems. 
The Coulomb interaction between the electrons (say $V$) can be
viewed as a perturbation of a ``free Hamiltonian'' modeling the interaction
with the nuclei (in the Born-Oppenheimer approximation). One can take
advantage of the particular form of $V$ (which is a linear combination of
products of two creation and two annihilation operators in the second 
quantization picture) to represent the perturbative expansions using a 
given family of Goldstone diagrams (see e.g. ref.~\onlinecite{Morita} 
for such a family). 
However, the general results on perturbative
expansions (such as the convergence of the time-dependent wave operator) 
do not depend on such a particular choice. 

Thus, we consider an Hamiltonian
$H(t)$ and we build its evolution operator $U(t,t_0)$, which is 
the solution of the Schr\"odinger equation
(in units $\hbar=1$)
\begin{eqnarray}
\imath\frac{\partial U(t,t_0)}{\partial t} &=&  H(t) U(t,t_0),
\label{dUdt}
\end{eqnarray}
with the boundary condition $U(t_0,t_0)=1$.
In perturbation theory, 
$H(t) \defeq \ee^{-\epsilon |t|} \ee^{\imath H_0 t} V \ee^{-\imath H_0t}$
is the adiabatically switched interaction Hamiltonian in the 
interaction picture (here $H_0$ and $V$ stand respectively for 
the ``free'' and interaction terms of the initial Hamiltonian)
and singularities show up in the adiabatic limit
($t_0\to -\infty$ and $\epsilon\to0$).
Morita discovered~\cite{Morita} that, in this setting, 
the \emph{time-dependent wave operator} 
\begin{eqnarray*}
\UGL(t,t_0) &\defeq & U(t,t_0) P \big(P U(t,t_0) P\big)^{-1},
\end{eqnarray*}
where $P$ is the projection onto the model space $M$,
has no singularity in the adiabatic limit. Moreover,
the wave operator determines the effective Hamiltonian because
\begin{eqnarray}
\Heff &\defeq & \lim_{\epsilon\to0} P H \UGL(0,-\infty).
\label{Heff}
\end{eqnarray}
However, as we have already alluded to, the effective computation of these
operators raises several combinatorial and analytical problems that have been 
addressed in a long series of articles 
(several of which will be referred to in the present article).

In the first sections of the paper, we consider 
a general time-dependent Hamiltonian $H(t)$ 
(not necessarily in the interaction picture).
In this broader setting, Jolicard~\cite{Jolicard-89} found that
the time-dependent wave operator 
provides also a powerful description of
the evolution of quantum
systems (see ref.~\onlinecite{Jolicard-03} for applications).
Then, we derive three (rigorously proven) series expansions of the
wave operator. The first one is classical and can be physically
interpreted as the replacement of causality
(i.e. the Heaviside step function $\theta(t-t')$)
by a ``propagator'' 
$\theta_P(t-t')\defeq\theta(t-t') Q - \theta(t'-t) P$,
where $Q=1-P$. This ``propagator'' is causal out of the
model space ($\theta_P(t-t')=0$ for $t<t'$ on the image of $Q$) and anticausal 
on it, like the Feynman propagator of quantum field theory~\cite{Pauli}.
However, this sum of causal and anticausal orderings
is cumbersome to use in practice. A second series
expansion is obtained by writing the wave operator
as a sum of integrals over all possible time orderings
of the Hamiltonians $H(t_i)$ (see Sect.~\ref{permutation-sect}). This expansion,
parametrized by all the permutations (or equivalent families), is
used in many-body theory and gives rise to a large
number of complicated terms. 
The third expansion is obtained by noticing that some
time orderings can be added to give simpler expressions.
This series is naturally indexed by trees and is the
main new tool developed in the present paper.
Among others, we derive a very simple recurrence relation
for the terms of the series.
We also show that the very structure of the corresponding
generalized iterated integrals showing up in the expansion 
is interesting on its own. These integrals carry naturally a rich algebraic 
structure that is connected to several recent results in the field of combinatorial
Hopf algebras and noncommutative symmetric functions. The corresponding algebraic results 
that point out in the direction of the existence of a specific Lie theory for effective Hamiltonians (generalizing the usual Lie theory) are gathered 
in an Appendix.

In the last sections of the paper, we restrict $H(t)$ to
the interaction picture and we consider the 
adiabatic limit.
We first prove that the adiabatic limit exists 
non perturbatively. We show that the effective Hamiltonian
defined by eq.~(\ref{Heff}) has the expected properties.
Then, we expand the series and we
give a rigorous (but lengthy) proof that the term
corresponding to each time ordering has an adiabatic limit.
Then, we consider the series indexed by trees and
we give a short and easy proof of the existence of that limit.
Finally, we provide a direct rule to calculate the term corresponding
to a given tree and establish the connection between
the time-dependent approach and the time-independent
equations discovered by Lindgren~\cite{Lindgren74}
and Kvasni{\v{c}}ka~\cite{Kvasnicka-74}.

The existence of this series indexed by trees 
can be useful in many ways: (i) It
describes a sort of superstructure that is common
to all many-body theories without knowing the
exact form of the interaction Hamiltonian;
(ii) It considerably simplifies the manipulation of the
general term of the series by providing a powerful
recurrence relation;
(iii) It provides simple algorithms to calculate
the terms of the series;
(iv) The number of trees of order $n$, $\frac{1}{n+1}
{2n\choose n}\approx \frac{4^n}{n^{\frac{3}{2}}\sqrt\pi}$ being 
subexponential, it improves the convergence of the 
series~\cite{Gurau,Maassen};
(v) It can deal with problems where the Hamiltonian $H_0$ 
is not quadratic.
Indeed, many-body theories most often require the Hamiltonian $H_0$
to be free, i.e. to be a quadratic function of the 
fields\cite{quadratic}.
As noticed by Bulaevskii~\cite{Bulaevskii}, this is not
the good point of view for some applications. For example, in the
microscopic theory of the interaction of radiation with matter,
it is natural to take for $H_0$ the Hamiltonian describing
electrons and nuclei in Coulomb interaction~\cite{DelSole-84},
the perturbation being the interaction with the transverse 
electric field. In that case, quadratic free Hamiltonians many-body theories break down
whereas our approach is still valid.
Actually, it is precisely for that reason that we originally developed
the tree series approach; (vi) Last, but not least, the tree-theoretical
approach connects many-body theories with a large field of knowledge that
originates in the ``birth of modern bijective combinatorics'' in the 
seventies with
in particular the seminal works of Foata, Sch\"utzenberger and Viennot
\cite{Foata,Viennot-81}. See e.g. ref.~\onlinecite{Leroux} for a survey 
of the modern combinatorial theory of tree-like structures.

From the physical point of view, the tree expansion
is particularly interesting in the adiabatic limit.
Indeed, the denominator of each of its terms
is a product of $E_i^Q-E_j^P$ factors, where
$E_j^P$ is the energy of a state in the model 
space $M$ and 
$E_j^Q$ the energy of a state not belonging to $M$.
In the usual many-body expansions,
the denominators are products of
$\sum_i E_i^Q-\sum_j E_j^P$ factors, where
the sums contain various numbers of elements (corresponding therefore 
to multiple transitions between low-energy and excitated levels).
In that sense, the tree expansion
is the simplest possible because each
term is a product of single transitions between
two states.

We now list the main new results of this paper:
(i) A recursion formula that generates the simplified
terms of 
the time-dependent perturbation series (theorem~\ref{thmOmegaT});
(ii) when the interaction is adiabatically switched on,
a non-perturbative proof of the convergence of the wave operator
and a characterization of the states of the model space that
are transformed into eigenstates of $H$ by the wave operator
(theorem~\ref{Nonperthm});
(iii) a proof of the existence of the adiabatic limit
for the terms of the series expansion of the wave operator
(theorem~\ref{convpermut});
(iv) a recursive formula (lemma~\ref{lemmaneuf})
and an explicit form (theorem~\ref{thmRS})
for the general term of the time-independent
perturbation series.

\section{Time-dependent Hamiltonian and combinatorics}
\label{tdH}
We consider a time-dependent Hamiltonian $H(t)$,
which is a self-adjoint operator on a Hilbert space
$\calH $, and its evolution operator $U(t,t_0)$ defined
in eq.~(\ref{dUdt}).
Since we are interested in the combinatorial
aspects of the problem, we consider the simple case where
$H(t)$ is a strongly continuous map from
$\mathbb{R}$ into the bounded self-adjoint operators
on $\calH$~\cite{relativelybounded}.
In that case, the Picard-Dyson series
\begin{eqnarray*}
1 + \sum_{n=1}^\infty (-\imath)^n 
   \int_{t_0}^{t} \dd t_1 
   \int_{t_0}^{t_1} \dd t_2
    \dots 
   \int_{t_0}^{t_{n-1}} \dd t_n
   H(t_1)\dots H(t_n),
\end{eqnarray*}
converges in the uniform operator topology to $U(t,t_0)$
and $U(t,t_0)$ is a jointly strongly continuous
two-parameter family of unitaries on $\calH$
(see section X.12 of ref.~\onlinecite{ReedSimonII}).

Following Morita~\cite{Morita},
Jolicard~\cite{Jolicard-89}
established a connection between the evolution 
operator and the effective Hamiltonian approach by defining
\begin{eqnarray*}
\UGL(t,t_0) &\defeq & U(t,t_0) P \big(P U(t,t_0) P\big)^{-1},
\end{eqnarray*}
where $P$ is a projection operator onto $M$ and
$ \big(P U(t,t_0) P\big)^{-1}$ is the inverse of
$P U(t,t_0) P$ as a map from 
$M=P\calH$ to itself. 
This map is invertible if and only if there is
no state $|\phi\rangle$ in $M$ such that
$\langle \phi| U(t,t_0) P=0$.
This condition is similar to the one
of time-independent perturbation theory~\cite{Bloch}.
We assume from now on that the condition is satisfied and 
we define three expansions for $\Omega (t,t_0)$.

\subsection{First expression for $\UGL$}
We start by proving an elegant expression for $\UGL$,
that was stated by Michels and Suttorp~\cite{Michels2}
and Dmitriev and Solnyshkina~\cite{Dmitriev3}.
\begin{thm}
\begin{eqnarray}
\UGL(t,t_0) &=& P + Q
 \sum_{n=1}^\infty (-\imath)^n 
   \int_{t_0}^{t} \dd t_1 
   \int_{t_0}^{t} \dd t_2
    \dots 
   \int_{t_0}^{t} \dd t_n 
   H(t_1)\theta_P(t_1-t_2) H(t_2)\dots \theta_P(t_{n-1}-t_n) H(t_n)P,
\label{MS}
\end{eqnarray}
where $Q=1-P$ and $\theta_P(t)=\theta(t)Q - \theta(-t) P$,
with $\theta$ the Heaviside step function.
\end{thm}
\begin{proof}
We first rewrite the Picard-Dyson series as
$U(t,t_0)=1+\sum_n U_n(t,t_0)$ with
$U_1(t,t_0)\defeq -\imath \int_{t_0}^t \dd t_1 H(t_1)$ and, for $n>1$,
\begin{eqnarray*}
U_n(t,t_0) &\defeq &
 (-\imath)^n 
   \int_{t_0}^{t} \dd t_1 
    \dots 
   \int_{t_0}^{t} \dd t_n
   H(t_1)\theta(t_1-t_2) \dots \theta(t_{n-1}-t_n) H(t_n).
\end{eqnarray*}
Then, by using $\theta(t)+\theta(-t)=1$, we notice that
\begin{eqnarray*}
\theta(t) &=& P + \theta(t)-\theta(t)P-\theta(-t)P
  = P + \theta_P(t).
\end{eqnarray*}
Now, we replace $\theta(t)$ by the sum of
operators $P + \theta_P(t)$
in the expression for $U_n(t,t_0)$. This gives us
$2^{n-1}$ terms with various numbers of $P$ and
$\theta_P$. 
Denote by $C_n(t,t_0)$ the term with no $P$
(with the particular case $C_1(t,t_0)=U_1(t,t_0)$).
Take then any other term. There is an index $i$ such that
the first $P$ from the left occurs after $H(t_i)$. Therefore,
the integrand of this term is
\begin{eqnarray*}
H(t_1)\theta_P(t_1-t_2) \dots \theta_P(t_{i-1}-t_i) H(t_i)
P H(t_{i+1})\dots
\end{eqnarray*}
Observe that the integral over $t_1,\dots,t_i$ is
independent from the integral over $t_{i+1},\dots,t_n$.
The first integral gives $C_i(t,t_0)$, the second
integral is a term of the Picard-Dyson series for $U_{n-i}(t,t_0)$.
Thus, the sum of the $2^{n-1}$ terms yields
\begin{eqnarray*}
U_n(t,t_0) &=& C_n(t,t_0) + \sum_{i=1}^{n-1}
  C_i(t,t_0) P U_{n-i}(t,t_0).
\end{eqnarray*}
If we denote by $K(t,t_0)$ the sum of all the $C_n(t,t_0)$
with $n>0$, we obtain
$U = 1 + K + KP(U-1)$, so that
\begin{eqnarray}
UP &=& P + KPUP.
\label{KPUP}
\end{eqnarray}
The operator $K$ is called the \emph{reduced evolution
operator} by Lindgren and collaborators~\cite{Lindgren3}.
If we define $\omega \defeq P + QKP$, eq.~(\ref{KPUP}) becomes
\begin{eqnarray*}
UP &=& P + (\omega-P+PKP)PUP
= P + \omega PUP -PUP +PKPUP.
\end{eqnarray*}
This equation can be simplified by using eq.~(\ref{KPUP}) again
\begin{eqnarray*}
UP &=& P + \omega PUP -PP=\omega PUP.
\end{eqnarray*}
Thus, $\omega=\UGL$ and eq.~(\ref{MS}) is satisfied.
\end{proof}

Despite its elegance, eq.~(\ref{MS}) is not
immediately usable. To illustrate this point,
consider the third-order term
\begin{eqnarray*}
\UGL_3 &=& \imath Q
   \int_{t_0}^{t} \dd t_1 
   \int_{t_0}^{t} \dd t_2
   \int_{t_0}^{t} \dd t_3 
   H(t_1)\theta_P(t_1-t_2) H(t_2)\theta_P(t_2-t_3) H(t_3)P.
\end{eqnarray*}
If we expand $\theta_P(t)=\theta(t)Q - \theta(-t) P$,
we obtain four terms
\begin{eqnarray*}
\imath QH(t_1)QH(t_2)QH(t_3)P &&
   \mathrm{for}\,\,t_1\ge t_2\,\,\mathrm{and}\,\,t_2\ge t_3,\\
-\imath QH(t_1)QH(t_2)PH(t_3)P &&
   \mathrm{for}\,\,t_1\ge t_2\,\,\mathrm{and}\,\,t_2\le t_3,\\
-\imath QH(t_1)PH(t_2)QH(t_3)P &&
   \mathrm{for}\,\,t_1\le t_2\,\,\mathrm{and}\,\,t_2\ge t_3,\\
\imath QH(t_1)PH(t_2)PH(t_3)P &&
   \mathrm{for}\,\,t_1\le t_2\,\,\mathrm{and}\,\,t_2\le t_3.
\end{eqnarray*}
The first and last terms have integration
range $t_1\ge t_2 \ge t_3$ and 
$t_3\ge t_2 \ge t_1$, respectively and give rise to iterated integrals. 
The integration range of the second term 
is $t_1\ge t_2$ and $t_2\le t_3$.
Such an integration range is not convenient
because the relative position of $t_1$ and
$t_3$ is not specified. The integration range
has to be split into the two subranges
$t_1\ge t_3 \ge t_2$ and 
$t_3\ge t_1 \ge t_2$.
Each subrange defines now an iterated integral.
For example $t_1\ge t_3 \ge t_2$ gives
\begin{eqnarray*}
-\imath   \int_{t_0}^{t} \dd t_1 
   \int_{t_0}^{t_1} \dd t_3
   \int_{t_0}^{t_3} \dd t_2 
QH(t_1)QH(t_2)PH(t_3)P.
\end{eqnarray*}
Similarly, the integration range of the third term
($t_1\le t_2$ and $t_2\ge t_3$) is the union of
$t_2\ge t_1 \ge t_3$ and 
$t_2\ge t_3 \ge t_1$.
We see that $\Omega_3$ is sum of six iterated integrals
corresponding to the six possible orderings of
$t_1$, $t_2$ and $t_3$.

\subsection{$\UGL$ in terms of permutations}
\label{permutation-sect}
We consider again the previous example,
and we change variables to have a fixed 
integration range $s_1\ge s_2 \ge s_3$.
If we sum over all time orderings, we obtain
\begin{eqnarray}
\UGL_3 &=& \imath
\int_{t_0}^{t} \dd s_1 
   \int_{t_0}^{s_1} \dd s_2
   \int_{t_0}^{s_2} \dd s_3  
  \Big(
  QH(s_1) Q H(s_2) Q H(s_3)P
  -
  QH(s_1) Q H(s_3) P H(s_2)P
  -
  QH(s_2) Q H(s_3) P H(s_1)P
\nonumber\\&&
  -
  QH(s_2) P H(s_1) Q H(s_3)P
  -
  QH(s_3) P H(s_1) Q H(s_2)P
  +
  QH(s_3) P H(s_2) P H(s_1)P
  \Big).
\label{Omega3}
\end{eqnarray}

In ref.~\onlinecite{BP-09}, we showed that this result can be generalized
to all orders and that $\Omega_n$ is a sum of $n!$
iterated integrals corresponding to all the orderings of
$t_1,\dots,t_n$. More precisely, we obtained
the series expansion for $\UGL$
\begin{eqnarray} 
\UGL(t,t_0) &=&
 P + 
  \sum_{n=1}^\infty 
  \sum_{\sigma\in \calS_n}
   \int_{t_0}^{t} \dd t_1 
   \int_{t_0}^{t_1} \dd t_2
    \dots 
   \int_{t_0}^{t_{n-1}} \dd t_n
   X(t_{\sigma(1)})\dots X(t_{\sigma(n)}),
\label{Omegasigma}
\end{eqnarray} 
where $\calS_n$ is the group of permutations of $n$ elements.
The operators $X$ are defined, for $n=1$, by
$X(t)\defeq -\imath Q H(t) P$
and, for $n>1$ and any $\sigma\in \calS_n$, by
$X(t_{\sigma(1)}) \defeq -\imath Q H(t_{\sigma(1)})$ and
\begin{eqnarray*}
  X(t_{\sigma(p)}) &\defeq & -\imath Q H(t_{\sigma(p)})
  \,\,\mathrm{if}\,\, 1<p<n\,\,\mathrm{and}\,\,
  \sigma(p)>\sigma(p-1),\\
  X(t_{\sigma(p)}) &\defeq & \imath P H(t_{\sigma(p)})
  \,\,\mathrm{if}\,\, 1<p<n\,\,\mathrm{and}\,\,
  \sigma(p)<\sigma(p-1),\\
  X(t_{\sigma(n)}) &\defeq & -\imath Q H(t_{\sigma(n)}) P
  \,\,\mathrm{if}\,\, \sigma(n)>\sigma(n-1),\\
  X(t_{\sigma(n)}) &\defeq & \imath P H(t_{\sigma(n)}) P
  \,\,\mathrm{if}\,\, \sigma(n)<\sigma(n-1).
\end{eqnarray*}

Each term of eq.~(\ref{Omegasigma}) is now written as
an iterated integral.
However, the expansion (\ref{Omegasigma})
is still not optimal because some of its terms
can be summed to get a simpler expression.
As an example, consider the fourth and fifth terms
of eq.~(\ref{Omega3}), where we replace $s_i$ by $t_i$.
\begin{eqnarray*}
Z &\defeq & -\imath
\int_{t_0}^{t} \dd t_1 
   \int_{t_0}^{t_1} \dd t_2
   \int_{t_0}^{t_2} \dd t_3 
  \Big(
  QH(t_2) P H(t_1) Q H(t_3)P
  +
  QH(t_3) P H(t_1) Q H(t_2)P
  \Big).
\end{eqnarray*}
The first and second terms of the right hand
side of this equation are denoted by $Z_1$
and $Z_2$, respectively.
We transform $Z_2$
by exchanging variables $t_2$ and $t_3$.
\begin{eqnarray*}
Z_2 &=& -\imath
\int_{t_0}^{t} \dd t_1 
   \int_{t_0}^{t_1} \dd t_3
   \int_{t_0}^{t_3} \dd t_2 
  QH(t_2) P H(t_1) Q H(t_3)P
=
-\imath \int_{t_0}^{t} \dd t_1 
   \int_{t_0}^{t_1} \dd t_2
   \int_{t_2}^{t_1} \dd t_3 
  QH(t_2) P H(t_1) Q H(t_3)P,
\end{eqnarray*}
where we also exchanged the order of the integrations
over $t_3$ and $t_2$.
This can be added to $Z_1$
and we obtain the simpler expression
\begin{eqnarray*}
Z &=& -\imath
\int_{t_0}^{t} \dd t_1 
   \int_{t_0}^{t_1} \dd t_2
   \int_{t_0}^{t_1} \dd t_3 
  QH(t_2) P H(t_1) Q H(t_3)P.
\end{eqnarray*}
Such a simplification is not possible for the
other terms of $\UGL_3$.
In the next section, we determine how this simplification
can be extended to the general term of $\UGL$. 

Before closing this section, we need to specify more precisely the 
relation between
the permutations $\sigma$ and the sequence of $P$ and $Q$ in the expansion 
of eq.~(\ref{MS}). When we expand all the $\theta_P(t_i-t_{i+1})$
in eq.~(\ref{MS}), we obtain an integrand of the
form 
$(-\imath)^n QH(t_1) R_1 H(t_2)\dots R_{n-1} H(t_n) P$, multiplied by a product of 
Heaviside functions, where $R_i$ takes the value $-P$ or $Q$.
We aim to determine the relation between the sequence
$R_1\dots R_{n-1}$ and the permutations $\sigma$ 
in eq.~(\ref{Omegasigma}).
From the definition of $X(t_{\sigma(i)})$, it appears that
$R_i = -P$ if $\sigma(i)>\sigma(i+1)$ and 
$R_i = Q$ if $\sigma(i)<\sigma(i+1)$. The set of indices $i$ such that
$\sigma(i)>\sigma(i+1)$ is called the \emph{descent set} of $\sigma$,
denoted by $D_\sigma$. It is also called the \emph{shape} of the permutation~\cite{Viennot}.
For instance, the descent set of the permutations $(213)$ and $(312)$ is
$\{1\}$, corresponding to $(R_1, R_2) = (-P, Q)$.

\subsection{Permutations and trees}
\label{permutation-tree}
In many-body physics, the expansion in Goldstone diagrams corresponds
(among other things) to the expansion of $\Omega$ into all
time orderings $\Omega_\sigma$.
In that context, several authors noticed that 
some diagrams corresponding to different orderings can
be added to give a simple 
sum~\cite{KuoOsnes,Michels2,Frantz,Bethe-63,Brandow,Olszewski-04},
as we saw in the previous section.
These are special cases of the simplification that we shall
present which, as far as we know, was never stated in full generality.
The first difficulty is to find the proper combinatorial
object to represent the sets of permutations 
that lead to simplified sums.
We shall find it after a short algebraic detour meant to recall the 
notion of tree and its relation to permutations~\cite{Foata,Viennot}.
The trees we consider have received various names in the literature:
binary trees in ref.~\onlinecite{Viennot}, but also (to quote 
only a few occurrences) plane rooted complete binary tree~\cite{Simion-00},
extended binary trees~\cite{Panayotopoulos}
or planar binary trees~\cite{LodayRonco}. 
Since these objects are rarely used in physics, we
first link them with the trees commonly found in graph theory.
This lengthy definition will then be replaced by a much easier one.

\begin{figure}
\begin{center}
\includegraphics{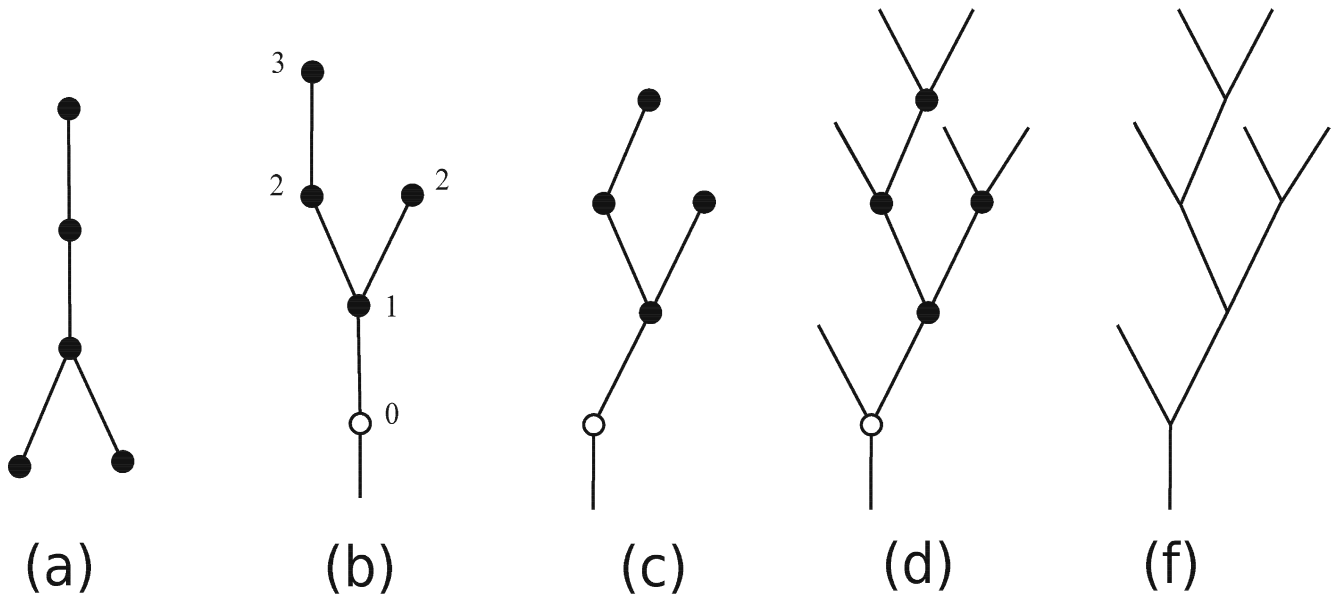}\end{center}
\caption{Construction of plane rooted complete binary tree. (a) a common tree;
(b) a rooted tree with its vertex levels; (c) a plane rooted binary
tree; (d) a plane rooted complete binary tree;
(f) a simplified plane rooted complete binary tree.}
\label{fig:PBT}
\end{figure}

A \emph{common tree} is a connected graph without loops. In other words, 
a common tree is a set of vertices linked by edges such that there
is a unique way to go from any vertex to any another one by
travelling along edges. An example of a common tree is given in fig.~\ref{fig:PBT}(a).
A \emph{rooted tree} is a common tree where one vertex was selected. This particular vertex
is called the \emph{root}. The \emph{level} of a vertex in a rooted tree is the number
of edges that separates this vertex from the root. The root has level 0.
It is natural to draw a rooted tree by putting the root at the bottom and
drawing a vertex above another one if its level is larger. 
The rooted tree of fig.~\ref{fig:PBT}(b) was obtained by selecting as the root the
lowest left vertex of fig.~\ref{fig:PBT}(a). The root is indicated by a white dot
and a dangling line.
In a rooted tree, the \emph{children} of a vertex $v$ are the vertices 
$v'$ that are linked to $v$ by an edge and such that the level of 
$v'$ is larger than the level of $v$.
A \emph{plane rooted binary tree} is a rooted tree where each vertex has zero, 
one or two children, and each edge is oriented either to the left or to the 
right. If a vertex has two children, then one of the edges is oriented to the 
left and the other one to the right.
Fig.~\ref{fig:PBT}(c) shows one of the plane rooted trees that can be obtained from the rooted tree of
fig.~\ref{fig:PBT}(b). The adjective ``plane'' means that an edge going to the
right cannot be identified with, or deformed into an edge going to the left.
A \emph{plane rooted complete binary tree}
is a plane rooted binary tree where each vertex is ``completed'' by drawing
\emph{leaves} as follows:
if a vertex has no child, draw a leaf (i.e. a line) to the left and one to the right,
if a vertex has one child, then draw a leaf to the right if the child is to the left
and draw it to the left if the child is to the right.
Fig.~\ref{fig:PBT}(d) shows the plane rooted complete binary tree that is obtained from
the plane rooted binary tree of fig.~\ref{fig:PBT}(c).
In practice, the vertices are no longer necessary and they are not drawn, as in
fig.~\ref{fig:PBT}(f).
If $Y_n$ denotes the set of plane rooted complete binary trees with $n$
vertices, we see that
$Y_1=\{\Y\}$,
$Y_2=\{\deuxun,\deuxdeux\}$,
$Y_3=\{\troisun,\troisdeux,\troistrois,\troisquatre,\troiscinq\}$.
They are much more numerous than the common trees
(there is only one common tree with one, two or three vertices).
For notational convenience, 
plane rooted complete binary trees will be simply called ``trees''
in the rest of the paper.

Fortunately, there exists a much simpler definition of the trees, that
use a sort of building rule.
We first denote the \emph{empty tree}, i.e. the tree with no vertex,
by$\|$, which is a dangling line without root (a dangling line with a root
and no other vertex belongs to the tree $\Y$).
Then, for any integer $n$, $Y_n$,
is defined recursively by $Y_0 \defeq\{\|\}$ and, for $n>0$,
$Y_n \defeq \{ T_1 \vee T_2 \,:\, T_1 \in Y_k, T_2 \in Y_{n-k-1},\, 
k=0,\dots,n-1\}$,
where $T_1 \vee T_2$ is the grafting of the two trees
$T_1$ and $T_2$, by which the dangling lines of $T_1$ and $T_2$
are brought together and a new root (with its own dangling line) is grown at their juncture.
For example, $\|\vee\|=\Y$,
$\|\vee\Y=\deuxdeux$, $\Y \vee \Y= \troistrois$.
Note that each tree of $Y_n$ has $n$ vertices (including a root)
and $n+1$ leaves. The order $|T|$ of a tree $T$ is the number of its vertices.

If $C_n$ denotes the number of elements of $Y_n$, the recursive
definition of $Y_n$ implies that $C_0=1$ and
$C_n=\sum_{k=0}^{n-1} C_k C_{n-k-1}$, so that
$C_n=\frac{1}{n+1}\binom{2n}{n}$ are the famous Catalan numbers.
For $n$=0 to 10, 
$C_n$=1, 1, 2, 5, 14, 42, 132, 429, 1430, 4862, 16796.
The Catalan numbers enumerate a large number of (sometimes quite different) 
combinatorial objects~\cite{StanleyII}. 
The main practical interest of trees with respect to other
combinatorial interpretations is that their recursive
definition is very easy to implement. 

We noticed in the previous section that, at order three, the terms $\Omega_\sigma$
corresponding to two specific permutations $\sigma$ can be added
to give a simple result. 
At the general order, we shall see that the sum of $\Omega_\sigma$
is simple if it is carried out
over permutations $\sigma$ associated with the same tree.
But for this, we need to associate a tree to each permutation.
The relevant map from permutations to trees
belongs to ``the ABC's of classical 
enumeration''~\cite{Rawlings} and is historically one of the founding 
tools of modern bijective and algebraic 
combinatorics~\cite{Foata,Viennot,StanleyI}.
We describe it in the next section.


\subsubsection{From permutations to trees}
\label{perm-tree}
We first map any $n$-tuple 
$I=(i_1,\dots,i_n)$ of distinct integers
to a tree $T\in Y_n$. The mapping $\phi$ is
defined recursively as follows:
if $I=(i)$ contains a single integer, then
$\phi(I)\defeq\Y$; otherwise, pick up the smallest
element $i_k$ of $I$, then
\begin{eqnarray*}
\phi(I) &\defeq & \|\vee\phi\big((i_2,\dots,i_n)\big), 
\,\,\mathrm{if}\,\,k=1,\\
\phi(I)&\defeq &\phi\big((i_1,\dots,i_{n-1})\big)\vee \|,
\,\,\mathrm{if}\,\,k=n,\\
\phi(I)&\defeq &\phi\big((i_1,\dots,i_{k-1})\big)\vee 
   \phi\big((i_{k+1},\dots,i_n)\big),
\,\,\mathrm{otherwise}.
\end{eqnarray*}
In the following, we frequently abuse notation by writing
$\phi(i_1,\dots,i_n)$ or even $\phi(i_1 \dots i_n)$ instead of
$\phi\big((i_1,\dots,i_n)\big)$.
For a permutation $\sigma\in \calS_n$,
the corresponding tree is
$\phi(\sigma)=\phi\big(\sigma(1),\dots,\sigma(n)\big)$.
For example
\begin{eqnarray*}
\phi(1) &=& \Y,\\
\phi(12) &=& \deuxdeux,\quad
\phi(21) = \deuxun ,\\
\phi(123) &=& \troiscinq, \quad
\phi(132) = \troisquatre,\quad
\phi(213) = \troistrois,\\
\phi(231) &=& \troisdeux,\quad
\phi(312) = \troistrois,\quad \phi(321) = \troisun.
\end{eqnarray*}
Note that the two permutations 
$(213)$ and $(312)$ correspond to
the same tree, and they are also
the two permutations that add up
to a simple sum in the calculation
of $Z$ at the end of section~\ref{permutation-sect}.
This is not a coincidence.

To simplify the proofs, we embrace the three cases of the definition
of $\phi$ into a single one as follows. We first define the
\emph{concatenation product} of two 
tuples $I=(i_1,\dots,i_n)$ and $J=(j_1,\dots,j_m)$
as $I\cdot J=(i_1,\dots,i_n,j_1,\dots,j_m)$.
We extend this definition to the case of the
zero-tuple $I_0=\emptyset$ by $I_0\cdot I=I\cdot I_0=I$.
Then, for any $n$-tuple $I$ of distinct integers, 
we define $\phi(I)$ by $\phi(I)=\|$ if $n=0$ and
$\phi(I)=\phi(I_1)\vee\phi(I_2)$ otherwise,
where $I_1$ and $I_2$ are determined by
$I=I_1 \cdot (\min I) \cdot I_2$.
Note that $I_1$ or $I_2$ can be the zero-tuple.

We first prove an easy lemma.
\begin{lem}
If the elements of the two $n$-tuples $I=(i_1,\dots,i_n)$ and 
$J=(j_1,\dots,j_n)$ of distinct integers have the same ordering
(i.e.  if $i_k < i_l$ if and only if $j_k<j_l$ for all
$k$ and $l$ in $\{1,\dots,n\}$), then 
$\phi(I)=\phi(J)$.
\end{lem}
\begin{proof}
The proof is by induction.
If $n=0$, then $I=J=\emptyset$ and
$\phi(I)=\phi(J)=\|$. Assume that the property is true
for $k$-tuples of distinct integers up to $k=n-1$
and take two $n$-tuples $I$ and $J$ having the same ordering.
Then, the minimum element of both is at the same position $k$
(i.e. $\min I = i_k$ and $\min J = j_k$) and 
$I=I_1\cdot (i_k) \cdot I_2$,
$J=J_1\cdot (j_k) \cdot J_2$, where
$I_1$ and $J_1$ ($I_2$ and $J_2$, respectively)
are two $(k-1)$-tuples
($(n-k)$-tuples, respectively)
of distinct integers have the same ordering.
By the recursion hypothesis, we have
$\phi(I_1)=\phi(J_1)$ and $\phi(I_2)=\phi(J_2)$ and 
the definition of $\phi$ gives us
$\phi(I)=\phi(I_1)\vee\phi(I_2)=\phi(J_1)\vee\phi(J_2)=\phi(J)$.
\end{proof}

As a useful particular case, we consider the situation
where $J$ describes the ordering of the elements of $I$: 
we order the elements of $I=(i_1,\dots,i_n)$ increasingly
as $i_{k_1} < \dots < i_{k_n}$. Then $j_l$ is the position
of $i_l$ in this ordering. More formally, 
$J\defeq (\tau^{-1}(1),\dots,\tau^{-1}(n))$, 
where $\tau$ is the  permutation $(k_1,\dots,k_n)$.
The $n$-tuple $J$ is called the \emph{standardization} of
$I$ and it is denoted by $st(I)$.
If we take the example of 
$I=(5,8,2)$, the position of 5, 8 and 2 in the ordering 
$2 < 5 < 8$ is 2, 3
and 1, respectively. Thus, $st(5, 8, 2) = (2, 3, 1)$.
By construction, $I$ and $st(I)$ have the same ordering
and $\phi(I)=\phi(st(I))$.
We extend the standardization to the case
of $I=\emptyset$ by $st(\emptyset)=\emptyset$.

\subsubsection{From trees to permutations}
Conversely, we shall need to know the permutations 
corresponding to a given tree:
$S_T\defeq \{\sigma\in \calS_{|T|}\,:\, \phi(\sigma)=T\}$
(we extend this definition to the case of
$T=\|$ by defining the zero-element permutation group
$\calS_0 \defeq \{\emptyset\}$).
The solution of this problem is given by
\begin{lem}
If $T=T_1\vee T_2$, where $|T_1|=n$ and $|T_2|=m$
($n$ or $m$ can be zero), all the permutations of $S_T$ have the
form $I=I_1\cdot (1)\cdot I_2$,
where $I_1$ is a subset of $n$ elements of $\{2,\dots,n+m+1\}$
ordered according to a permutation $\alpha$ of
$S_{T_1}$ (i.e. $st(I_1)=\alpha$)
and $I_2$ is the complement of $I_1$ in 
$\{2,\dots,n+m+1\}$, 
ordered according to a permutation $\beta$ of
$S_{T_2}$ (i.e. $st(I_2)=\beta$).
\end{lem}
\begin{proof}
The proof is given in refs.~\onlinecite{Panayotopoulos}
and \onlinecite{LodayRonco},
but we can sketch it here for completeness.
The simplest examples are
$S_T=\{\emptyset\}$ for $T=\|$ and
$S_T=\{(1)\}$ for $T=\Y$.
Now take $T=T_1\vee T_2$ as in the lemma.
By the definition of $\phi$, the minimum of the tuple
$I=(\sigma(1),\dots,\sigma(n+m+1))$
is $\sigma(n+1)=1$ and
$I=I_1\cdot (1)\cdot I_2$, where
$\phi(I_1)=T_1$ and $\phi(I_2)=T_2$.
We saw in the previous section that
$\phi(I_1)=\phi(st(I_1))$. By definition
$st(I_1)$ is a permutation of $\calS_n$. Therefore,
$st(I_1)$ belongs to $S_{T_1}$ and, similarly,
$st(I_2)$ belongs to $S_{T_2}$.
It is now enough to check that each element of 
$S_T$ is obtained exactly once by running the construction
over all orderings and all permutations of $T_1$ and $T_2$.
\end{proof}
This lemma allows us to recursively determine
the number of elements of $S_T$, denoted by $|S_T|$, by
$|S_T|=1$ for $T=\|$ and $T=\Y$ and, for $T=T_1\vee T_2$,
\begin{eqnarray}
|S_T| &=& \binom{|T|-1}{|T_1|} |S_{T_1}|\,|S_{T_2}|.
\label{nbST}
\end{eqnarray}
See ref.~\onlinecite{Panayotopoulos} for an alternative approach.

Example:
Consider the tree $T=T_1\vee T_2$, with $T_1=\Y$
and $T_2=\troistrois$, so that $n=1$ and $m=3$.
$S_{T_1}$ contains the single permutation $\alpha=(1)$ and,
according to the examples given in the previous section,
the two permutations of $S_{T_2}$ are
$\beta_1=(213)$ and $\beta_2=(312)$.
We choose the permutations $\alpha$ and $\beta_1$,
we pick up $n=1$ element (for example $3$) in the set $J=\{2,3,4,5\}$,
so that $I_1=(3)$ and we order the remaining
elements $\{2,4,5\}$ according to $\beta_1$, so
that $I_2=(4,2,5)$.
This gives us $\sigma=(31425)$.
If we pick up the other elements of $J$ to build $I_1$
we obtain $(21435),(41325)$ and $(51324)$.
We add the elements obtained by choosing $\beta_2$
and we obtain eight permutations:
$$S_T=
\{(21435),(21534),(31425),(31524),(41325),(41523),(51324),(51423)\}.$$
We can check that eq.~(\ref{nbST}) holds and that $|S_T|=8$.

\subsection{Recursion formula}
The permutations corresponding to a tree can be used
to make a partial summation of the terms of the Picard-Dyson expansion.
\begin{dfn}
For any tree $T$, we define $\UGL_T(t,t_0)$ by
$\UGL_T(t,t_0)=P$ if $T=\|$ and
\begin{eqnarray*}
\UGL_T(t,t_0) &\defeq &  \sum_{\sigma \in S_T}
   \int_{t_0}^{t} \dd t_1 
   \int_{t_0}^{t_1} \dd t_2
    \dots 
   \int_{t_0}^{t_{n-1}} \dd t_n
   X(t_{\sigma(1)})\dots X(t_{\sigma(n)}),
\end{eqnarray*}
otherwise, where $n=|T|$.
\end{dfn}
With this notation we have obviously
$\UGL(t,t_0)=\sum_T \UGL_T(t,t_0)=\sum_\sigma\UGL_\sigma(t,t_0)$, with the 
notation $\UGL_\sigma(t,t_0) \defeq \int_{t_0}^{t} \dd t_1 
   \int_{t_0}^{t_1} \dd t_2
    \dots 
   \int_{t_0}^{t_{n-1}} \dd t_n
   X(t_{\sigma(1)})\dots X(t_{\sigma(n)})$
and where $\sigma$ runs over all permutations (of all orders).
The term of order 0 of this series is $\Omega_|=P$
and the term of order one is
\begin{eqnarray*}
\UGL_T(t,t_0) &=& 
  -\imath\int_{t_0}^t \dd s Q H(s) P,
\end{eqnarray*}
for $T=\Y$.
The other terms enjoy a remarkably simple
recurrence relation:
\begin{thm}
\label{thmOmegaT}
If $|T|>1$, then $\Omega_T(t,t_0)$
can be expressed recursively by
\begin{eqnarray*}
\UGL_T(t,t_0) &=& 
  -\imath\int_{t_0}^t \dd s Q H(s) \UGL_{T_2}(s,t_0),
  \,\,\,
  \mathrm{if}\,\, T=\|\vee T_2,\nonumber\\
\UGL_T(t,t_0) &=& 
  \imath\int_{t_0}^t \dd s \UGL_{T_1}(s,t_0)H(s)P,
  \,\,\,
  \mathrm{if}\,\, T=T_1\vee\|,\nonumber\\
\UGL_T(t,t_0) &=& 
  \imath\int_{t_0}^t \dd s \UGL_{T_1}(s,t_0)H(s)\UGL_{T_2}(s,t_0),
  \,\,\,
  \mathrm{if}\,\, T=T_1\vee T_2,
\end{eqnarray*}
where $T_1\not=\|$ and $T_2\not=\|$.
\label{OmegaT}
\end{thm}
Note that a similar recursive expression was conjectured by
Olszewski for the nondegenerate Rayleigh-Schr\"odinger 
expansion~\cite{Olszewski-04}.

\begin{proof}
Let us prove the theorem recursively.
Consider an arbitrary $T=T_1\vee T_2$, $|T|>1$, and
assume the formulas to hold for all the trees $T'$ with $|T'|<|T|$.
Consider for example the case where $T_1\not=\|$ and $T_2\not=\|$ 
(the other cases are even simpler). We define:
$$A_T \defeq \imath\int_{t_0}^t \dd s \UGL_{T_1}(s,t_0)H(s)\UGL_{T_2}(s,t_0)
=\sum\limits_{\alpha\in S_{T_1},\beta\in S_{T_2}}\imath\int_{t_0}^t 
\dd s\UGL_{\alpha}(s,t_0)H(s)\UGL_{\beta}(s,t_0)
.$$
This first important point is that, for a given tree
$T$, all the permutations $\sigma\in S_T$ have the
same descent set. This is a well-known 
fact~\cite{Viennot,LodayRonco}
that can be deduced from the characterization of $S_T$
at the end of section \ref{permutation-tree}.
As a consequence, the sequence of operators
$P$, $Q$ and $H$ is the same for all
$\alpha$ and $\beta$ in $A_T$, and only the
order of the arguments $t_i$ varies.
Therefore, the conditions of
lemma~\ref{fundam} (see appendix \ref{crucial})
are satisfied and we get:
$$A_T=\sum\limits_\gamma\UGL_\gamma(t,t_0),$$
where $\gamma$ runs over the permutations such that $\gamma(|T_1|+1)=1$, 
$st(\gamma(1),...,\gamma(|T_1|))\in S_{T_1}$, 
$st(\gamma(|T_1|+2),...,\gamma(|T_1|+|T_2|+1))\in S_{T_2}$.
The set of permutations $\gamma$ satisfying these
equations is precisely $S_T$,
so that, finally: $A_T=\UGL_T(t,t_0)$. This concludes the proof of the theorem.
\end{proof}

\subsection{Remarks}
\subsubsection{Nonlinear integral equation}

If we denote $\chi(t,t_0)=\UGL(t,t_0)-P$, then the recurrence relations
add up to
\begin{eqnarray}
\imath \chi(t,t_0) &=&
  \int_{t_0}^t \dd s Q H(s) P
  +
  \int_{t_0}^t \dd s Q H(s) \chi(s,t_0)
  -
  \int_{t_0}^t \dd s \chi(s,t_0)H(s)P
  -
  \int_{t_0}^t \dd s \chi(s,t_0)H(s)\chi(s,t_0).
\label{ichi}
\end{eqnarray}
The derivative of this equation with respect to $t$
was obtained in a different way by Jolicard~\cite{Jolicard-89}.

\subsubsection{Permutations, trees and descents}
We saw that, for a given tree $T$, all permutations
of $S_T$ have the same descent set. 
We can now give more details~\cite{LodayRonco,LodayRonco2}.
The relation between the trees and the sequences
of operators $P$ and $Q$ in eq.~(\ref{MS}) is very simple.
Consider the sequence of leaves from left to right.
Each leaf pointing to the right corresponds to
a $P$, each leaf pointing to the left correspond to a $Q$.
For example, the tree $\troisdeux$
corresponds to the sequence $QQPP$. 
From the combinatorial point of view, this description emphasizes the 
existence of a relationship 
between trees and descent sets (or, equivalently, hypercubes), see e.g.
refs.~\onlinecite{Viennot,Stanley,Chapoton-00-1,Chapoton-00-2,LodayRonco2} and 
our Appendix.

\section{Adiabatic switching}
Morita's formula is most often applied to
an interaction Hamiltonian 
$H^\epsilon(t)\defeq\ee^{-\epsilon |t|}\ee^{\imath H_0 t}V\ee^{-\imath H_0t}$.
We write $E_0,...,E_n,...$ and $\Phi_0,...,\Phi_n,...$ for the eigenvalues 
of $H_0$ and for an orthogonal basis of corresponding eigenstates.
We assume that the spectrum is discrete and that the eigenvalues are ordered 
by (weakly) increasing order. The ground state may be degenerate 
($E_0=E_1=...=E_k$ for a given $k$).
The model space $M$ (see the Introduction) is the vector space generated
by the lowest $N$ eigenstates of $H_0$ (with $N\ge k$).
We assume that the energies of the eigenstates of $M$
are separated by a finite gap from 
the energies of the eigenstates that do not belong to $M$.
The projector $P$ is the projector onto the model space $M$.
Following the notation of the Introduction,
the energies of the eigenstates that belong (resp. do not belong) to $M$
are denoted by $E^P_i$ (resp. $E^Q_i$).

In this section, we prove the convergence of
each term of the perturbation expansion of the wave operator
when $\epsilon\to0$. For notational convenience,
we assume that $t\le 0$.
We first give a nonperturbative proof of this convergence.
Then, we expand in series and we consider the different 
cases of the previous sections.

\subsection{Nonperturbative proof}
The nonperturbative proof is important in this context because
its range of validity is wider than the series expansion
(no convergence criterion for the series is required). 
Moreover, the proof that the wave operator indeed leads to
an effective Hamiltonian is much easier to give in the
nonperturbative setting.

The first condition required in the nonperturbative setting is that the
perturbation $V$ must be relatively
bounded with respect to $H_0$ with a bound strictly
smaller than 1. This condition is satisfied for the Hamiltonian
describing nuclei and electrons interacting through
a Coulomb potential~\cite{Kato}.
Before stating the second condition,
we define the time independent Hamiltonian
$h(\lambda)=H_0+\lambda V$, its eigenvalues 
$E_j(\lambda)$ and its eigenprojectors 
$P_j(\lambda)$, such that 
$h(\lambda) P_j(\lambda)=E_j(\lambda) P_j(\lambda)$.
The second condition is that the eigenvalues $E_j(\lambda)$ coming from 
the eigenstates of the model space (i.e. such that
$P_j(0) P=P_j(0)$) are separated by a finite gap from the
rest of the spectrum.
According to Kato~\cite{Kato}, the eigenvalues and eigenprojectors 
can be chosen analytic in $\lambda$.
Then, a recent version of the adiabatic theorem~\cite{BPS,BPS-PRL} shows
that there exists a unitary operator $A$, independent
of $\epsilon$, such that
\begin{eqnarray*}
\lim_{\epsilon\to0} || U_\epsilon(0,-\infty)P_j(0) -
      \ee^{\imath\theta_j/\epsilon} A P_j(0)|| &=& 0,
\end{eqnarray*}
where $U_\epsilon(t,t_0)$ is the evolution operator
for the Hamiltonian
$H^\epsilon(t)=\ee^{-\epsilon |t|} \ee^{\imath H_0 t} V \ee^{-\imath H_0t}$
and
\begin{eqnarray*}
\theta_j &=& \int_{0}^1 \frac{E_j(0)-E_j(\lambda)}{\lambda} \dd \lambda.
\end{eqnarray*}
In other words, the singularity of 
$ U_\epsilon(0,-\infty)P_j(0) $ is entirely
described by the factor 
      $\ee^{\imath\theta_j/\epsilon}$.
The operator $A$ satisfies the intertwining property
$AP_j(0)=P_j(1) A$. 
To connect this result with the case that
we are considering in this paper, we 
have to choose a model space $M$
that satisfies the following condition:
there is a set $I$ of indices $j$
such that $P=\sum_{j\in I} P_j(0)$,
where $P$ is the projector onto $M$.

This enables us to give a more precise condition for
the invertibility of $P U_\epsilon(0,-\infty) P$.
The adiabatic theorem shows that, for small enough $\epsilon$,
$P U_\epsilon(0,-\infty) P$ is invertible if and only if
$PAP$ is invertible.
If we rewrite $PAP=\sum_j P A P_j(0)=\sum_j P P_j(1) A$,
the unitarity of $A$ implies that $PAP$ is invertible
iff the kernel of $\sum_j P P_j(1)$ is trivial. 
We recover the well-known invertibility condition~\cite{Bloch}
that no state of the model space should be orthogonal 
to the vector space spanned by all the eigenstates of $H$ with 
energy $E_j(1)$, where $j$ runs over $I$.
Note that, when the condition of invertibility is not
satisfied, it can be recovered by adding the perturbation
step by step~\cite{BPS}.

Then, we have
\begin{thm}
\label{Nonperthm}
With the given conditions, the wave operator
\begin{eqnarray*}
\WO &\defeq & 
\lim_{\epsilon\to0} 
U_\epsilon(0,-\infty) P (P U_\epsilon(0,-\infty) P)^{-1}
\end{eqnarray*}
is well defined.
Moreover, there are states 
$|\tildephi_j\rangle$ in the model space
such that $\WO|\tildephi_j\rangle$ 
is an eigenstate of $H$ with eigenvalue
$E_j(1)$
and the effective Hamiltonian
$\Heff\defeq PH\WO$ satisfies
$\Heff |\tildephi_j\rangle =
 E_j(1) |\tildephi_j\rangle$.
\end{thm}
\begin{proof}
We first define $A_{jk}=P_j(0) A P_k(0)$, for $j$ and
$k$ in $I$. 
Then an inverse $B$ of $P A P$ 
is defined by 
$\sum_{k\in I} A_{jk} B_{kl}=\delta_{jl} P_j(0)$
where $B_{jk}=P_j(0) B P_k(0)$.
Then, $(P U_\epsilon(0,-\infty) P)^{-1}
\simeq \sum_{jk} \ee^{-\imath\theta_j/\epsilon} B_{jk}$
and 
\begin{eqnarray*}
\Omega_\epsilon(0,-\infty) & \defeq &
U_\epsilon(0,-\infty) P (P U_\epsilon(0,-\infty) P)^{-1}
\simeq \sum_{jk} A P_j(0) B_{jk}.
\end{eqnarray*}
Since the right hand side does not depend on 
$\epsilon$, then
$\Omega_\epsilon(0,-\infty)$ has no singularity at 
$\epsilon=0$ and
\begin{eqnarray*}
\WO &= & 
\lim_{\epsilon\to0} \Omega_\epsilon(0,-\infty) = 
\sum_{jk} A P_j(0) B_{jk}.
\end{eqnarray*}
This proves the existence of the wave operator.
To prove the existence of the states
$|\tildephi_j\rangle$ of the theorem, define
$|\tildephi_j\rangle=PA |\varphi_j\rangle$,
where 
$|\varphi_j\rangle$ is an eigenstate of
$P_j(0)$:
$P_j(0)|\varphi_j\rangle=|\varphi_j\rangle$.
Indeed, we have
\begin{eqnarray*}
\WO|\tildephi_j\rangle &=&
\WO P A P_j(0)|\varphi_j\rangle =
\sum_{km} A P_k(0) B_{km} A_{mj}|\varphi_j\rangle
= A P_j(0)|\varphi_j\rangle = P_j(1) A |\varphi_j\rangle,
\end{eqnarray*}
where we used the intertwining property
in the last equation.
We can now check that $\WO|\tildephi_j\rangle$ is an eigenstate of $H$
with eigenvalue $E_j(1)$. 
\begin{eqnarray}
H \WO|\tildephi_j\rangle &=&
h(1) P_j(1) A |\phi_j\rangle
= E_j(1) P_j(1) A |\phi_j\rangle
= E_j(1)  \WO|\tildephi_j\rangle.
\label{HOmega}
\end{eqnarray}
Finally, by multiplying 
eq.~(\ref{HOmega}) by $P$ from the left, we obtain
\begin{eqnarray*}
\Heff |\tildephi_j\rangle &=&
 E_j(1) P \WO | \tildephi_j\rangle =
 E_j(1) |\tildephi_j\rangle,
\end{eqnarray*}
because $P\WO=P$ and $P|\tildephi_j\rangle=|\tildephi_j\rangle$. 
\end{proof}
Thus, the eigenvalues of $\Heff$
are eigenvalues of the full Hamiltonian $H$.
This is exactly what is expected from an effective Hamiltonian.
In practice, the operator $A$ is not known and the states
$|\tildephi_j\rangle$ are obtained by diagonalizing $\Heff$.

\subsection{Series expansion}
We consider again the series expansion in terms of permutations.
A straightforward calculation~\cite{Goldstone,Gross} of the
Picard-Dyson series gives us
\begin{eqnarray*}
U_\epsilon(0,-\infty) |\Phi_0\rangle &=& |\Phi_0\rangle + 
   \sum_{n=1}^\infty \sum_{i_1\dots i_n}
   \frac{|\Phi_{i_1}\rangle\langle\Phi_{i_1}| V |\Phi_{i_{2}}\rangle
   \dots
   \langle\Phi_{i_{n-1}}| V |\Phi_{i_{n}} \rangle
   \langle\Phi_{i_n}| V |\Phi_{0} \rangle}
   {(E_0-E_{i_1}+n\imath\epsilon)(E_0-E_{i_2}+(n-1)\imath\epsilon)
   \dots
   (E_0-E_{i_n}+\imath\epsilon) } ,
\end{eqnarray*}
where we used the
completeness relation $1=\sum_i |\Phi_{i}\rangle \langle\Phi_{i}|$.
This expression clearly shows that the 
terms of the expansion (and the
evolution operator) are divergent 
as $\epsilon\to0$ when any $E_{i_k}$ is equal to $E_0$.

For $\sigma\in \calS_n$, we set 
$\Omega_\sigma(t)\defeq \Omega_\sigma(t,-\infty )$. We then  have
\begin{eqnarray*}
\Omega_\sigma(t) &=& (-\imath)^n
   \int_{-\infty}^{t} \dd t_1 
   \int_{-\infty}^{t_1} \dd t_2
    \dots 
   \int_{-\infty}^{t_{n-1}} \dd t_n
  Q \ee^{(\epsilon +\imath H_0)t_{\sigma(1)}} V 
  \ee^{-\imath H_0 t_{\sigma(1)}} R^1_\sigma 
\\&&
   \ee^{(\epsilon +\imath H_0)t_{\sigma(2)}} V 
  \ee^{-\imath H_0 t_{\sigma(2)}} R^2_\sigma 
  \dots
   R^{n-1}_\sigma 
  \ee^{(\epsilon +\imath H_0)t_{\sigma(n)}} V 
  \ee^{-\imath H_0 t_{\sigma(n)}} P,
\end{eqnarray*}
where $R^k_\sigma\defeq Q$ if $\sigma(k+1)>\sigma(k)$ and 
$R^k_\sigma\defeq -P$ if $\sigma(k+1)<\sigma(k)$.
We replace $R^k_\sigma$ by $\pm\sum_{\alpha_{k+1}} 
|\alpha_{k+1}\rangle\langle\alpha_{k+1}|$
where, if $R^k_\sigma=Q$, then $\pm=+$ and
the sum is over the image of $Q$, and
if $R^k_\sigma=-P$, then $\pm=-$ and
the sum is over the image of $P$.
Thus
\begin{eqnarray}
\Omega_\sigma(t) &=& (-\imath)^n (-1)^d
   \int_{-\infty}^{t} \dd t_1 
   \int_{-\infty}^{t_1} \dd t_2
    \dots 
   \int_{-\infty}^{t_{n-1}} \dd t_n
  \ee^{(\epsilon +\imath F_1-\imath F_2)t_{\sigma(1)}} 
  \ee^{(\epsilon +\imath F_2-\imath F_3)t_{\sigma(2)}} 
  \dots
  \ee^{(\epsilon +\imath F_n-\imath F_{n+1})t_{\sigma(n)}} 
\nonumber\\ &&
\sum_{\alpha_1 \dots \alpha_{n+1}}
|\alpha_1\rangle\langle\alpha_1 |V|\alpha_2\rangle 
  \dots
\langle\alpha_n |V|\alpha_{n+1}\rangle \langle \alpha_{n+1}|,
 \label{omegasigma}
\end{eqnarray}
where $d$ is the number of elements of the descent set of $\sigma$, 
$F_i$ is the energy of $\alpha_i$ and where the sum over $\alpha_1$
is over the image of $Q$, the sum over $\alpha_{n+1}$ is over the
image of $P$ and the sum over $\alpha_k$ for $1<k<n+1$
is over the image of $Q$ if $\sigma(k)>\sigma(k-1)$ and over the image of $P$ 
otherwise.
Consider now the time integral
\begin{eqnarray*}
f_\sigma(t) &\defeq & 
   \int_{-\infty}^{t} \dd t_1 
   \int_{-\infty}^{t_1} \dd t_2
    \dots 
   \int_{-\infty}^{t_{n-1}} \dd t_n
  \ee^{(\epsilon +\imath F_1-\imath F_2)t_{\sigma(1)}} 
  \ee^{(\epsilon +\imath F_2-\imath F_3)t_{\sigma(2)}} 
  \dots
  \ee^{(\epsilon +\imath F_n-\imath F_{n+1})t_{\sigma(n)}} 
\\&=&
   \int_{-\infty}^{t} \dd s_{\tau(1)}
   \int_{-\infty}^{s_{\tau(1)}} \dd s_{\tau(2)}
    \dots 
   \int_{-\infty}^{s_{\tau(n-1)}} \dd s_{\tau(n)}
  \ee^{(\epsilon +\imath F_1-\imath F_2)s_1} 
  \ee^{(\epsilon +\imath F_2-\imath F_3)s_2} 
  \dots
  \ee^{(\epsilon +\imath F_n-\imath F_{n+1})s_n}, 
\end{eqnarray*}
where $\tau=\sigma^{-1}$.
The integral over $s_{\tau(n)}$ is
\begin{eqnarray*}
   \int_{-\infty}^{s_{\tau(n-1)}} \dd s_{\tau(n)}
  \ee^{(\epsilon +\imath F_{\tau(n)}-\imath F_{\tau(n)+1})s_{\tau(n)}} 
&=&
  \frac{\ee^{(\epsilon +\imath F_{\tau(n)}-\imath F_{\tau(n)+1})s_{\tau(n-1)}}}
  {(\epsilon +\imath F_{\tau(n)}-\imath F_{\tau(n)+1})}.
\end{eqnarray*}
The integrand of the integral over $s_{\tau(n-1)}$
becomes
\begin{eqnarray*}
\frac{\ee^{(2\epsilon +\imath (F_{\tau(n)}+F_{\tau(n-1)}-
   F_{\tau(n)+1}-F_{\tau(n-1)+1})s_{\tau(n-1)}}}
  {(\epsilon +\imath F_{\tau(n)}-\imath F_{\tau(n)+1})}.
\end{eqnarray*}
A straightforward recursive argument shows that
\begin{eqnarray*}
f_\sigma(t) &=& 
\frac{\ee^{X_\sigma(n) t}}
  {X_\sigma(1)\dots X_\sigma(n)},
\end{eqnarray*}
where
\begin{eqnarray}
X_\sigma(k) &\defeq &
  k \epsilon + \imath (F_{\sigma^{-1}(n)} + \dots + F_{\sigma^{-1}(n-k+1)}
  - F_{\sigma^{-1}(n)+1} - \dots - F_{\sigma^{-1}(n-k+1)+1}).
\label{defXsigma}
\end{eqnarray}
Therefore,
\begin{eqnarray}
\Omega_\sigma(t) &=& 
\sum_{\alpha_1 \dots \alpha_{n+1}}\frac{(-\imath)^n (-1)^d\ee^{X_\sigma(n) t}}
  {X_\sigma(1)\dots X_\sigma(n)}
|\alpha_1\rangle\langle\alpha_1 |V|\alpha_2\rangle 
  \dots
\langle\alpha_n |V|\alpha_{n+1}\rangle \langle \alpha_{n+1}|.
\label{Omegasigmat}
\end{eqnarray}

\subsection{Examples}
A few examples of $\Omega_\sigma(0)$ are
\begin{eqnarray*}
\Omega_{(1)}(0) &=& 
  (-\imath) \sum_{\Phi_i\Phi_j}
\frac{|\Phi_i\rangle\langle\Phi_i|V|\Phi_j\rangle
  \langle \Phi_j|}{\epsilon + \imath (E^Q_i-E^P_j)},\\
\Omega_{(12)}(0) &=& 
  (-\imath)^2 \sum_{\Phi_i\Phi_j\Phi_k}
\frac{|\Phi_i\rangle\langle\Phi_i|V|\Phi_j\rangle
         \langle\Phi_j|V|\Phi_k\rangle
  \langle \Phi_k|}
  {\big(\epsilon + \imath (E^Q_j-E^P_k)\big)
  \big(2\epsilon + \imath (E^Q_i-E^P_k)\big)},\\
\Omega_{(21)}(0) &=& 
  -(-\imath)^2 \sum_{\Phi_i\Phi_j\Phi_k}
\frac{|\Phi_i\rangle\langle\Phi_i|V|\Phi_j\rangle
         \langle\Phi_j|V|\Phi_k\rangle
  \langle \Phi_k|}
  {\big(\epsilon + \imath (E^Q_i-E^P_j)\big)
  \big(2\epsilon + \imath (E^Q_i-E^P_k)\big)}.
\end{eqnarray*}
Finally we consider two examples that will prove
useful:
\begin{eqnarray*}
\Omega_{(213)}(0) &=& 
  -(-\imath)^3 \sum_{\Phi_i\Phi_j\Phi_k\Phi_l}
\frac{|\Phi_i\rangle\langle\Phi_i|V|\Phi_j\rangle
         \langle\Phi_j|V|\Phi_k\rangle
         \langle\Phi_k|V|\Phi_l\rangle
  \langle \Phi_l|}
  {\big(\epsilon + \imath (E^Q_k-E^P_l)\big)
  \big(2\epsilon + \imath (E^Q_k+E^Q_i-E^P_l-E^P_j)\big)
  \big(3\epsilon + \imath (E^Q_i-E^P_l)\big)},\\
\Omega_{(312)}(0) &=& 
  -(-\imath)^3 \sum_{\Phi_i\Phi_j\Phi_k\Phi_l}
\frac{|\Phi_i\rangle\langle\Phi_i|V|\Phi_j\rangle
         \langle\Phi_j|V|\Phi_k\rangle
         \langle\Phi_k|V|\Phi_l\rangle
  \langle \Phi_l|}
  {\big(\epsilon + \imath (E^Q_i-E^P_j)\big)
  \big(2\epsilon + \imath (E^Q_k+E^Q_i-E^P_l-E^P_j)\big)
  \big(3\epsilon + \imath (E^Q_i-E^P_l)\big)}.
\end{eqnarray*}
By adding these two terms,
we obtain a denominator involving only the difference
of two energies.
\begin{eqnarray*}
\Omega_{(213)}(0) +
\Omega_{(312)}(0) &=& 
  -(-\imath)^3 \sum_{\Phi_i\Phi_j\Phi_k\Phi_l}
\frac{|\Phi_i\rangle\langle\Phi_i|V|\Phi_j\rangle
         \langle\Phi_j|V|\Phi_k\rangle
         \langle\Phi_k|V|\Phi_l\rangle
  \langle \Phi_l|}
  {\big(\epsilon + \imath (E^Q_i-E^P_j)\big)
  \big(\epsilon + \imath (E^Q_k-E^P_l)\big)
  \big(3\epsilon + \imath (E^Q_i-E^P_l)\big)}.
\end{eqnarray*}

Note that the sum is simpler than either
$\Omega_{(213)}(0)$ or $\Omega_{(312)}(0)$. This is
a general statement and the simplification becomes spectacular
at higher orders. For the example of $n=7$, there is a single tree
$T$ which is the sum of 80 permutations, and the denominator
of $\Omega_T$ is simpler than the denominator of $\Omega_\sigma$
for any of the 80 permutations $\sigma$ of $S_T$. This will
be proved in section~\ref{sect-trees}.
Note also that, if we assume that the states in the image of $P$
(i.e. the model space)
are separated from the states in the image of $Q$ by a finite
gap $\delta$, so that
$E^Q_i - E^P_j \ge \delta$, then the 
denominators of all the examples are non-zero
when $\epsilon\to 0$. In other words,
the limit $\lim_{\epsilon\to0} \Omega_\sigma(0)$
exists for all the examples. In the
next section, we show that this result is
true for all permutations $\sigma$.

\subsection{Convergence of $\Omega_\sigma(t)$}
Definition~(\ref{defXsigma}) is convenient for
a computer implementation but it does not make it clear
that $X_\sigma(k)$ is nonzero if $\epsilon=0$.
For that purpose, we need an alternative expression
for $X_\sigma(k)$, which is essentially a corrected version
of the graphical rule given by 
Michels and Suttorp~\cite{Michels2}.
We first extend any permutation $\sigma\in \calS_n$
to the sequence of $n+2$ integers
$\barsigma=(\barsigma(1),\dots,\barsigma(n+2))=(0,\sigma(1),\dots,\sigma(n),0)$.
Then, for $k\in\{1,\dots,n\}$, we define the two sets
\begin{eqnarray*}
S_\sigma^<(k) &\defeq & \{i\,|\, 1\le i \le n+1
  \,\,\mathrm{and}\,\, \barsigma(i) < k \le \barsigma(i+1)\},\\
S_\sigma^>(k) &\defeq & \{i\,|\, 1\le i \le n+1
  \,\,\mathrm{and}\,\, \barsigma(i) \ge  k > \barsigma(i+1)\}.
\end{eqnarray*}
For example, if $\sigma=(41325)$, then
$S_\sigma^<(1)=\{1\},
S_\sigma^<(2)=\{1,3\},
S_\sigma^<(3)=\{1,3,5\},
S_\sigma^<(4)=\{1,5\},
S_\sigma^<(5)=\{5\}$
and
$S_\sigma^>(1)=\{6\},
S_\sigma^>(2)=\{2,6\},
S_\sigma^>(3)=\{2,4,6\},
S_\sigma^>(4)=\{2,6\},
S_\sigma^>(5)=\{6\}$.
The graphical meaning of these sets is illustrated in
Figure \ref{fig:rule}. Notice that the vertical axis is oriented downwards in order to reflect the time-ordering in the integrals $S_\sigma(t)$.
                      
\begin{figure}
\begin{center}
\includegraphics{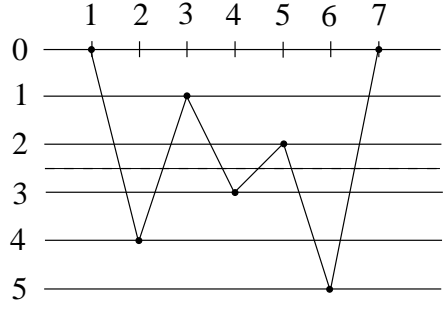}\end{center}
\caption{Construction of $S^<_\sigma(3)$ for 
  $\sigma=(41325)$. We build
$\bar\sigma=(0,4,1,3,2,5,0)$, we draw a continuous line
$L$ starting from $(1,\bar\sigma(1))=(1,0)$ to 
$(2,\bar\sigma(2))=(2,4)$,
to $(3,\bar\sigma(3))=(3,1)$,
\dots, up to
$(7,\bar\sigma(7))=(7,0)$. To determine 
$S^<_\sigma(3)$, we draw a horizontal dashed line just above
3 and we gather the segments of $L$ crossing the dashed line
from above.
In our case the segments are
$\big((1,0),(2,4)\big)$,
$\big((3,1),(4,3)\big)$ and
$\big((5,2),(6,5)\big)$.
$S^<_\sigma(3)$ is the set of abscissae of the first
point of each segment:
$S^<_\sigma(3)=\{1,3,5\}$.
Similarly,
$S^>_\sigma(3)$ is obtained from the segments that
cross the dashed line from below:
$S^>_\sigma(3)=\{2,4,6\}$.}
\label{fig:rule}
\end{figure}

\begin{lem}
(i)
$S^<_\sigma(k)$ cannot be empty and (ii)
$S^<_\sigma(k)$ and
$S^>_\sigma(k)$ have the same number of elements.
\end{lem}

The lemma follows from the graphical interpretation of the construction of 
$S_\sigma^>(k)$ and $S_\sigma^<(k)$. The graph of $\sigma$ (constructed as 
in Figure 1) is a sequence of edges connecting the points
$(i,\barsigma(i))$.
Since the graph starts from $(0,0)$ and since there exists one point with 
ordinate $n$, any horizontal line with non integer ordinate $y$, $0<y<n$, 
will be crossed from above by a segment (remember the vertical axis is 
oriented downwards). A similar elementary topological argument shows 
that such a horizontal line is always crossed successively from above 
and below by segments, the series of crossings starting from above and 
ending from below, which implies $|S^<_\sigma(k)|=|S^>_\sigma(k)|$.

The key step in the proof of convergence is 
\begin{lem}
For $\sigma\in \calS_n$ and a given sequence 
$\alpha_1,\dots,\alpha_{n+1}$, compatible with $\sigma$ 
(see eq.~(\ref{omegasigma})), we have
\begin{eqnarray}
X_\sigma(k) &=& k\epsilon + \sum_{j\in S_\sigma^<(n-k+1)} \imath F_j^Q
-\sum_{j\in S_\sigma^>(n-k+1)} \imath F_j^P,
\label{Xknew}
\end{eqnarray}
where we write $F_i^Q=F_i$ (resp. $F_i^P=F_i$) when $\alpha_i$ 
belongs to the image of $Q$ (resp. of $P$).
\end{lem}

\begin{proof}
We first show that it is true for $k=1$.
Indeed, 
$X_\sigma(1)=\epsilon + \imath(F_{\sigma^{-1}(n)}-F_{\sigma^{-1}(n)+1})$.
Let $j=\sigma^{-1}(n)$, we have
$\sigma(j)=n$.
Then, either 
$j=1$ and $|\alpha_1\rangle$ is in the image of $ Q$, or
$j>1$ and $\sigma\in \calS_n$ implies 
$\sigma(j)=n>\sigma(j-1)$, so that
$|\alpha_j\rangle$ is in the image of $ Q$. Thus, in all cases,
$F_{\sigma^{-1}(n)}= F_{\sigma^{-1}(n)}^Q$.
Consider now $F_{j+1}$. Either
$j=n$ and 
$|\alpha_{j+1}\rangle= |\alpha_{n+1}\rangle$ is in the image of $ P$,
or
$j<n$ and 
$\sigma(j+1)<\sigma(j)=n$, so that 
$|\alpha_{j+1}\rangle$ is in the image of $P$.
Thus, in all cases,
$F_{\sigma^{-1}(n)+1}= F_{\sigma^{-1}(n)+1}^P$.
Therefore, 
$X_\sigma(1)=\epsilon + \imath (F^Q_{\sigma^{-1}(n)}-
F^P_{\sigma^{-1}(n)+1})$.
On the other hand,
$S_\sigma^<(n)=\{\sigma^{-1}(n)\}$
and
$S_\sigma^>(n)=\{\sigma^{-1}(n)+1\}$ since $(j,n)$ is the only point of the graph with ordinate $n$. Thus,
the two members of eq.~(\ref{Xknew}) are equal for $k=1$.

Assume now that eq.~(\ref{Xknew}) holds for all the $X_\sigma(i)$ with
 $i=1,\dots,k$, $k<n$, and consider the equation (which is true by definition of the $X_\sigma(i)$s):
\begin{eqnarray}\label{rec}
X_\sigma(k+1) &=& X_\sigma(k)+\epsilon+ \imath (F_j- F_{j+1}),
\end{eqnarray}
where $j=\sigma^{-1}(n-k)$.
We first treat the case $1<j<n$.
Four possible situations can arise:
(i) $\sigma(j-1)>\sigma(j)>\sigma(j+1)$,
(ii) $\sigma(j-1)<\sigma(j)>\sigma(j+1)$,
(iii) $\sigma(j-1)>\sigma(j)<\sigma(j+1)$ and
(iv) $\sigma(j-1)<\sigma(j)<\sigma(j+1)$.
In case (i), we have 
$F_j=F^P_j$ and $F_{j+1}=F^P_{j+1}$.
On the other hand, condition (i) implies
$\barsigma(j)>\barsigma(j+1)>\barsigma(j+2)$,
so that $S_\sigma^>(n-k)=S_\sigma^>(n-k+1)$
and
$S_\sigma^<(n-k)$ is obtained from
$S_\sigma^<(n-k+1)$ by removing $\{j\}$
and adding $\{j+1\}$. 
eq.~(\ref{rec}) together with the hypothesis that eq.~(\ref{Xknew}) holds for $X_\sigma(k)$ 
imply that eq.~(\ref{Xknew}) holds for $X_\sigma(k+1)$.
Case (ii) implies
$F_j=F_j^Q$ and $F_{j+1}=F_{j+1}^P$,
case (iii) implies
$F_j=F_j^P$ and $F_{j+1}=F_{j+1}^Q$,
case (iv) implies
$F_j=F_j^Q$ and $F_{j+1}=F_{j+1}^Q$.
In all cases, these identities imply that the two expressions (\ref{rec}) and (\ref{Xknew}) for $X_\sigma(k+1)$ do agree.

It remains to treat the boundary cases.
If $j=1$, then $F_j=F^Q_1$ and we have
either (i) $\sigma(1)<\sigma(2)$ and
$F_2=F^Q_2$ or (ii) $\sigma(1)>\sigma(2)$ and
$F_2=F^P_2$. 
We know that $\barsigma(1)=0$, thus,
case (i) corresponds to
$\barsigma(1)<\barsigma(2)<\barsigma(3)$,
so that according to eq.~(\ref{Xknew}),
$$X_\sigma(k+1)-X_\sigma(k)=\imath (F_1^Q-F_2^Q),$$
in agreement with eq.~(\ref{rec}).
In case (ii) we have
$\barsigma(1)<\barsigma(2)>\barsigma(3)$,
which amounts to add $\imath F^Q_1$ and remove $\imath F^P_2$.
Again, eq.~(\ref{Xknew}) holds for $k+1$.
Finally, if $j=n$, then $F_{j+1}=F^P_{n+1}$ and we have
(i) $\sigma(n-1)<\sigma(n)$ and
$F_n=F^Q_n$ or (ii) $\sigma(n-1)>\sigma(n)$ and
$F_n=F^P_n$. 
Case (i) corresponds to
$\barsigma(n)<\barsigma(n+1)>\barsigma(n+2)$,
case (ii) corresponds to
$\barsigma(n)>\barsigma(n+1)>\barsigma(n+2)$.
In all cases, the relation given by
eq.~(\ref{Xknew}) is satisfied for
$k+1$ and the induction proof is complete.
\end{proof}

We can now ready to prove
\begin{thm}
\label{convpermut}
The limit
\begin{eqnarray}
\WO_\sigma &=& 
  \lim_{\epsilon\to0} \Omega_\sigma(0),
\end{eqnarray}
is well-defined.
\end{thm}
\begin{proof}
To prove the term-wise convergence
of $\UGL_\sigma(t)$ as $\epsilon\to 0$,
consider eq.~(\ref{Xknew}).
The sets  $S_\sigma^<(n-k+1)$
and $S_\sigma^>(n-k+1)$ have the
same number of elements, say $n_k$
and, by the gap hypothesis, we have
$F_j^Q - F_i^P \ge \delta$ for any $i$ and $j$. Therefore,
$|X_\sigma(k)|^2 \ge  k^2\epsilon^2 + 
  n_k^2 \delta^2 \ge \delta^2$,
since
$n_k\ge 1$. Thus, the denominator remains away
from zero by a finite amount for any $\epsilon\geq 0$
and the limit of $1/X_\sigma(k)$ for 
$\epsilon\to 0$ is well-defined.
\end{proof}

\section{Trees}
\label{sect-trees}
We showed that, for each permutation $\sigma$,
the wave operator $\Omega_\sigma(t,-\infty)$ has a well-defined 
limit as $\epsilon\to 0$. 
The detailed proof was rather lengthy and the final expression for 
$X_\sigma(k)$ suggests physically the simultaneous occurrence of 
transitions from states of the model space to states out of it.
The convergence is actually much easier to show in terms of trees, 
and the expressions showing up in the expansion are simpler, 
mathematically and physically, each factor of the denominator corresponding 
to a single difference between an energy in the model space and
an energy out of it.

In this section, if $N$ is the dimension of the model space, we write $i\in Q$ for $i>N$ and $j\in P$ 
for $1\leq j\leq N$, both for notational simplicity, and to emphasize 
the meaning of the indices, that correspond respectively to eigenstates 
in the image of $Q$ and in the image of $P$ (i.e. the model space). 

\begin{prop}
\label{OmegaTeps}
If $T=T_1\vee T_2$, then, for $t\le 0$,
\begin{eqnarray}
\Omega_T(t)\defeq\Omega_T(t,-\infty) &=& \sum_{i\in Q, j\in P}
  \ee^{(\imath E^Q_i-\imath E^P_j+ |T|\epsilon)t}
   \Omega^{ij}_T \,|\Phi_i\rangle\langle \Phi_j|,
\label{Omegaalphan}
\end{eqnarray}
where $\Omega^{ij}_T$ is obtained recursively by:

For $T=\Y$, $
\Omega^{ij}_{\Y} \defeq
-\imath \frac{\langle\Phi_i |V|\Phi_j\rangle}
   {\imath E_i^Q-\imath E_j^P+\epsilon}.$

For $T_1=\|$, $T_2\not=\|$: $
\Omega_T^{ij}  \defeq -\imath \sum\limits_{k \in Q} 
\frac{\langle\Phi_i |V|\Phi_k\rangle \Omega^{kj}_{T_2}}
  {\imath E_i^Q-\imath E_j^P+ |T|\epsilon}.$
  
For $T_1\not=\|$, $T_2=\|$:
$
\Omega_T^{ij}  \defeq \imath \sum\limits_{k \in P} 
  \frac{\Omega^{ik}_{T_1} \langle\Phi_k |V|\Phi_j\rangle}
  {\imath E_i^Q-\imath E_j^P+ |T|\epsilon}.$
  
For $T_1\not=\|$, $T_2\not=\|$:
$\Omega_T^{ij}  \defeq \imath \sum\limits_{k \in P, l\in Q} 
  \frac{\Omega^{ik}_{T_1} \langle\Phi_k |V|\Phi_l\rangle \Omega^{lj}_{T_2}}
  {\imath E_i^Q-\imath E_j^P+ |T|\epsilon}.$
  
\end{prop}

\begin{proof}
The computation of
 $\Omega^{ij}_{\Y}$ follows from eq.~(\ref{Omegasigmat}).
 Let us consider for example the case $T_1=\|$, $T_2\not=\|$.
Then, by applying theorem~\ref{OmegaT}:
\begin{eqnarray*}
\Omega_T(t) &=& -\imath \sum\limits_{i\in Q,j\in P} \int_{-\infty}^t \dd s
   Q \ee^{\epsilon s} 
  \ee^{-\imath E_i^Q s}
  \ee^{(\imath E_i^Q-\imath E_j^P+ |T_2|\epsilon)s}\Omega^{ij}_{T_2} 
  \ee^{iH_0s}V
  |\Phi_i\rangle \langle \Phi_j|.
\end{eqnarray*}
We replace $Q$ by $\sum\limits_{k\in Q} |\Phi_k\rangle\langle \Phi_k|$
and obtain, by using
$|T|=|T_2|+1$,
\begin{eqnarray*}
\Omega_T(t) &=& -\imath \sum\limits_{k,i\in Q, j\in P} \int_{-\infty}^t \dd s
  \ee^{(\imath E_k^Q-\imath E_j^P+ |T|\epsilon)s}
  \langle\Phi_k |V|\Phi_i\rangle \Omega^{ij}_{T_2} 
  |\Phi_k\rangle \langle \Phi_j|
\\&=&
 -\imath \sum\limits_{k,i\in Q, j\in P}
  \frac{\ee^{(\imath E_k^Q-\imath E_j^P+ |T|\epsilon)t}}
  {\imath E_k^Q-\imath E_j^P+ |T|\epsilon}
  \langle\Phi_k |V|\Phi_i\rangle \Omega^{ij}_{T_2} 
  |\Phi_k\rangle \langle \Phi_j|,
\end{eqnarray*}
The two other cases can be treated similarly.
\end{proof}

Since, for arbitrary $i$ and $j$, $E_i^Q- E_j^P \ge \delta$, we get:
\begin{cor}
\label{convtree}
The limit $\bar\Omega_T(t)=\lim\limits_{\epsilon\to 0}\Omega_T(t)$ is well defined.
\end{cor}
 
\subsection{Relation with the Rayleigh-Schr\"odinger perturbation theory}
Kvasni{\v{c}}ka~\cite{Kvasnicka-74} and Lindgren~\cite{Lindgren74}
independenty obtained an equation for the time-independent 
Rayleigh-Schr\"odinger perturbation theory of possibly degenerate systems:
\begin{eqnarray}
{[}\baromega,H_0{]}P &=& V \baromega P - \baromega P V \baromega P,
\label{Lindgren}
\end{eqnarray}
where the time-independent wave operator $\baromega P$
transforms eigenstates $|\Phi_0\rangle$ of $H_0$ into
eigenstates $\baromega P |\Phi_0\rangle$ of $H_0+V$
and where $P\baromega P=P$ (see ref.~\onlinecite{LindgrenMorrison}
p.~202 for details).

Equation~(\ref{Lindgren}) is an important generalization of 
Bloch's classical results~\cite{Bloch-58} because
it is also valid for a quasi-degenerate model space
(i.e. when the eigenstates of $H_0$ in the model
space have different energies).
The relation between time-dependent and time-independent
perturbation theory is established by the following
proposition:
\begin{prop}We have
$\baromega P=\WO=\lim_{\epsilon\to0}\Omega(0)$.
\end{prop}
\begin{proof}
We take the derivative of eq.~(\ref{ichi}) with respect to
time and we substitute $\chi(t,t_0)=\UGL(t,t_0)-P$.
This gives us
\begin{eqnarray*}
\imath\frac{\dd}{\dd t}\Omega(t,t_0) &=&
H(t) \Omega(t,t_0) - \Omega(t,t_0) H(t) \Omega(t,t_0).
\end{eqnarray*}
If we take $t=0$ and $t_0=-\infty$, we obtain by
continuity
\begin{eqnarray*}
\imath\frac{\dd}{\dd t}\Omega(t)|_{t=0} &=&
V \Omega(0) - \Omega(0) V \Omega(0).
\end{eqnarray*}
When we compare this equation with eq.~(\ref{Lindgren}),
we see that $\baromega P$ and $\WO$
satisfy the same equation if
\begin{eqnarray}
\imath \lim_{\epsilon\to0} \frac{\dd \Omega(t)}{\dd t}|_{t=0} &=& 
{[}\lim_{\epsilon\to0}\Omega(0),H_0{]}.
\label{dOmega}
\end{eqnarray}
To show this, we prove it for each term $\Omega_T$.
Indeed, eq.~(\ref{Omegaalphan}) gives us
\begin{eqnarray*}
\imath \frac{\dd \Omega_T(t)}{\dd t} &=& 
   \sum_{i\in Q, j\in P}
  (E_j^P-E_i^Q+\imath|T|\epsilon)
  \ee^{(\imath E_i^Q-\imath E_j^P+ |T|\epsilon)t}
   \Omega^{ij}_T \,|\Phi_i\rangle\langle \Phi_j|,
\end{eqnarray*}
and
\begin{eqnarray*}
{[}\Omega_T(t),H_0{]} &=& 
   \sum_{i\in Q, j\in P}
  (E_j^P-E_i^Q)
  \ee^{(\imath E_i^Q-\imath E_j^P+ |T|\epsilon)t}
   \Omega^{i j}_T \,|\Phi_i\rangle\langle \Phi_j|.
\end{eqnarray*}
By continuity in $\epsilon$, 
these two expressions are identical
for all $t$ when $\epsilon\to0$.
If we take $t=0$ and we sum over all trees $T$, then
we recover eq.~(\ref{dOmega}).
Therefore,
$\baromega P$ and $\WO$
satisfy the same equation. It remains to show
that they have the same boundary conditions:
$\WO P=\WO$
and $P\WO=P$.
By eq.~(\ref{MS}), these two equations are true for
$\Omega(t)$ with any value of $t$ and $\epsilon$.
\end{proof}

As a corollary, proposition \ref{OmegaTeps} 
provides a recursive construction of
the wave operator $\WO$.

\subsection{An explicit formula for $\Omega_T$}\label{sec:den}

In this section, we show how $\Omega_T(t)$ can be obtained 
non-recursively from the knowledge of $T$.
The key idea is to replace $T$ by another combinatorial object, better 
suited to that particular computation. 
We write therefore $\gamma_T$ for the smallest permutation for the 
lexicographical ordering in $S_T$ (we view a permutation as a word 
to make sense of the lexicographical ordering: to $(35421)$ corresponds 
the word $35421$, so that e.g. $(35421)<(54231)$). 
Since $S_T$ is always non empty, the map $\gamma:T\longmapsto \gamma_T$ 
is well-defined and an injection from the set of trees to the set 
of permutations. 
These permutations are called Catalan permutations~\cite{Panayotopoulos},
312-avoiding permutations (i.e. permutations for which there does not
exist $i<j<k$ such that $\sigma(j)<\sigma(k)<\sigma(i)$,
see ref.~\onlinecite{StanleyII} p.~224), 
Kempf elements~\cite{Postnikov-09}
or stack words~\cite{Feil-05}.

Various elementary manipulations can be done to understand such a map. 
We list briefly some obvious properties and introduce some notation that 
will be useful in our forthcoming developments.
If $I=(a_1,...,a_k)$ is a sequence of integers, we write $I[n]$ for the 
shifted sequence $(a_1+n,...,a_k+n)$. 

Then, let $T=T_1\vee T_2$ be a tree. The permutation $\gamma(T)$ 
(that we identify with the corresponding word or sequence) can be 
constructed recursively as 
$\gamma(\|)=\emptyset$, $\gamma(\Y)=(1)$ and
$$\gamma(T)\defeq(\gamma(T_1)[1],1,\gamma(T_2)[|T_1|+1])).$$
The left inverse of $\gamma$, say $\cal T$, is also easily described 
recursively as ${\cal T}(\emptyset)=\|$,
${\cal T}(1)\defeq \Y$ and
$${\cal T}(\sigma)\defeq {\cal T}((\sigma(1),...,\sigma(k))[-1])
\vee{\cal T}((\sigma(k+2),...,\sigma(n))[-k-1]),$$
where $\sigma\in \calS_n$, $\sigma=(\sigma(1),...,\sigma(k),1,\sigma(k+2),
...,\sigma(n))$ is in the image of $\gamma$.

Permutations in the image of $\gamma$ can be characterized recursively 
similarly: with the same notation as in the previous paragraph, 
$\sigma$ is in the image of $\gamma$ if and only if
$(\sigma(1),...,\sigma(k))[-1]$  and
$(\sigma(k+2),...,\sigma(n))[-k-1]$ are in the 
image of $\gamma$, where $k$ is the integer such that
$\sigma(k)=1$, 

We are now in a position to compute $\Omega_T(t)$. Let us write 
$\calA_T$ for all the sequences ${\bfalpha}=
(\alpha_1,...,\alpha_{n+1})$ associated to $\gamma(T)$ as in 
equation~(\ref{omegasigma}). We write, as usual, $F_i$ for the 
eigenvalue associated to $\alpha_i$. Recall that $\alpha_1\in Q$, 
$\alpha_{n+1}\in P$ whereas $\alpha_i, i\not=1,n+1$ is in $Q$ if 
$\sigma(i)>\sigma(i-1)$ and in $P$ otherwise. 

These sequences are actually common to the expansions of all the 
$\Omega_\sigma(t), \sigma\in S_T$ (this is because
 they depend only on the positions of descents in the 
permutations $\sigma\in S_T$, as discussed in the proof
of thm.~\ref{thmOmegaT}).
They appear therefore in the expansion of $\Omega_T(t)=\sum
\limits_{\sigma\in S_T}\Omega_\sigma(t)$.
They actually also correspond exactly to the sequences of eigenvectors 
that show up in the recursive expansion of $\Omega_T(t)$ 
(proposition~\ref{Omegaalphan}) (this should be clear from our previous 
remarks, 
but can be checked directly from the definition of the recursive expansion).
We can refine the recursion of proposition~\ref{OmegaTeps} accordingly:
\begin{lem}
\label{lemmaneuf}
We have: $\Omega_T(t)=\sum\limits_{\bfalpha}
\ee^{(\imath F_1-\imath F_{n+1}+|T|\epsilon)t}
\Omega_T^{{\bfalpha}}|\alpha_1\rangle\langle\alpha_{n+1}|$,
where $\Omega_T^{{\bfalpha}}(t)$ is defined recursively by:
$$\Omega_{\Y}^{\Phi_i,\Phi_j}=-\imath\frac{\langle\Phi_i|V|\Phi_j\rangle}
{\imath E_i^Q-\imath E_j^P+\epsilon},$$
$$\Omega_T^{{\bfalpha}}=\imath\frac{\Omega_{T_1}^{(\alpha_1,...,\alpha_k)}
\langle\alpha_k|V|\alpha_{k+1}\rangle\Omega_{T_2}^{(\alpha_{k+1},...,\alpha_n,
\alpha_{n+1})}}{\imath F_1-\imath F_{n+1}+|T|\epsilon},$$
where $T=T_1\vee T_2$ and $k=|T_1|$.
For $T_1=\|$, we have $\Omega^{(\alpha_1)}_{T_1}=-1$ and 
for $T_2=\|$, we have $\Omega^{(\alpha_{n+1})}_{T_2}=1$.
\end{lem}

Let us now consider $\gamma_T$. For $i=1,\dots,n$, we set: 
$l(i)=\inf\{j\leq i|\forall k,j\leq k\le i,\gamma_T(k)\ge \gamma_T(i)\}$ 
and $r(i)=\sup\{j\geq i|\forall l, j\geq l\ge i,\gamma_T(l)\ge\gamma_T(i)\}$. 
In words, $l(i)$ is defined as follows: consider all the consecutive
positions $k$ on the left of position $i$, such that the value of
the permutation $\gamma_T(k)$ is larger than $\gamma_T(i)$. Then,
$l(i)$ is the leftmost of these positions $k$.
Similarly, $r(i)$ is the rightmost position $j$ such that, on all
positions $k$ between $i$ and $j$, the permutation $\gamma_T(k)$ is 
larger than $\gamma_T(i)$.

\begin{thm}
\label{thmRS}
We have:
\begin{eqnarray*}
\Omega_T^{{\bfalpha}} &=& (-\imath)^n (-1)^{d-1}
\langle\alpha_1|V|\alpha_2\rangle...
\langle\alpha_n|V|\alpha_{n+1}\rangle\prod\limits_{i=1}^n
\frac{1}{\imath F_{l(i)}-\imath F_{r(i)+1}+(r(i)-l(i)+1)\epsilon},
\end{eqnarray*}
or, equivalently,
\begin{eqnarray*}
\Omega_T(t) &=& (-\imath)^n (-1)^{d-1}
\sum\limits_{\bfalpha}
|\alpha_1\rangle\langle\alpha_1|V|\alpha_2\rangle...
\langle\alpha_n|V|\alpha_{n+1}\rangle\langle\alpha_{n+1}|
\ee^{(\imath F_1-\imath F_{n+1}+|T|\epsilon)t}
\\&&
\prod\limits_{i=1}^n
\frac{1}
{\imath F_{l(i)}-\imath F_{r(i)+1}+(r(i)-l(i)+1)\epsilon},
\end{eqnarray*}
where $d$ is the number of leaves of $T$ pointing to the right.
\end{thm}
For example, if $T=\troistrois$, then
$\gamma_T=(213)$, $l=(113)$, $r=(133)$ and 
\begin{eqnarray*}
\Omega_T^{\Phi_i\Phi_j\Phi_k\Phi_l}
 &=& 
  -\imath \sum_{\Phi_i\Phi_j\Phi_k\Phi_l}
\frac{|\Phi_i\rangle\langle\Phi_i|V|\Phi_j\rangle
         \langle\Phi_j|V|\Phi_k\rangle
         \langle\Phi_k|V|\Phi_l\rangle
  \langle \Phi_l|}
  {\big(\epsilon + \imath (E^Q_i-E^P_j)\big)
  \big(\epsilon + \imath (E^Q_k-E^P_l)\big)
  \big(3\epsilon + \imath (E^Q_i-E^P_l)\big)}.
\end{eqnarray*}
\begin{proof}
We show that $\Omega^\alpha_T$ as defined in theorem~\ref{thmRS}
satisfies the recursion relation of lemma~\ref{lemmaneuf}.
Let us first consider that $T_1\not=\|$ and $T_2\not=\|$.
Let $i_0$ denote the index such that $\gamma_T(i_0)=1$,
with $1<i_0<n$. The term of the numerator corresponding to $i_0$
is $\langle\alpha_{i_0}|V|\alpha_{i_0+1}\rangle$, which is the
central term of the recursion relation.
We have $l(i_0)=1$ and $r(i_0)=n$. Thus, the denominator is 
$\imath F_{l(i_0)}-\imath F_{r(i_0)+1}+(r(i_0)-l(i_0)+1)\epsilon
=\imath F_1-\imath F_{n+1}+n\epsilon$, which is the denominator
of the recursion relation.
Now we check that the product of terms for $i<i_0$ in 
theorem~\ref{thmRS} is
$\Omega_{T_1}^{(\alpha_1,\dots,\alpha_{i_0-1})}$.
The matrix elements 
$\langle\alpha_1|V|\alpha_2\rangle...
\langle\alpha_{i_0-1}|V|\alpha_{i_0}\rangle$ obviously
agree, so we must check that the denominators agree.
Thus, we verify that, for $1\le i < i_0$,
$l_T(i)=l_{T_1}(i)$ and $r_T(i)=r_{T_1}(i)$, where
$l_T$ and $r_T$ denote the $l$ and $r$ vectors for tree $T$.
We know that $\gamma_T(i)=\gamma_{T_1}(i)+1$ for
$1\le i < i_0$. Thus, for $k\le i$, 
$\gamma_T(k)\ge \gamma_T(i)$ if and only if
$\gamma_{T_1}(k)\ge \gamma_{T_1}(i)$ and 
$l_T(i)=l_{T_1}(i)$.
For $r_T$, we notice that for $l=i_0$ we have
$\gamma_T(l)=1 < \gamma_T(i)$ and the relation
$\gamma_T(l) \ge \gamma_T(i)$ does not hold.
Therefore, $r_T(i)<i_0$ and 
$r_T(i)=r_{T_1}(i)$ by the same argument as
$l_T(i)=l_{T_1}(i)$.
The same reasoning holds for $T_2$ and the recursion relation
is satisfied.
The cases $T_1=\|$ or $T_2=\|$ are proved similarly.
\end{proof}
We conclude this section with a geometrical translation of
the previous theorem. 
Consider a tree $T$ with $|T|=n$ and number its leaves from 
1 for the leftmost leave to $n+1$ for the rightmost one.
For each vertex $v$ of $T$, take the subtree 
$T_v$ for which $v$ is the root. In other words,
$T_v$ is obtained by chopping the edge below $v$ and
considering the half-edge dangling from $v$ as the dangling line of 
the root of $T_v$.
For each tree $T_v$, build the pair $(l_v,r_v)$
which are the indices of the leftmost and rightmost leaves of 
$T_v$. Recall that $|\alpha_i\rangle$ belongs to the image
of $Q$ (resp. $P$) if leaf $i$ points to the left (resp. right): 
this implies in particular that $F_{l_v}=F_{l_v}^Q$ and $F_{r_v}=F_{r_v}^P$.
Then, the set of pairs $(l_v,r_v)$ where $v$ runs over the
vertices of $T$ is the same as the set of pairs
$(l(i),r(i)+1)$ of the theorem.
The formula for $\Omega_T^\alpha$ can be rewritten accordingly and this geometrical version can be proved recursively as 
theorem~\ref{thmRS}. Conversely, it can be used to determine the
tree corresponding to a given denominator.

\section{Conclusion}
We considered three expansions of the wave operator
and we proved their adiabatic convergence.
We proposed to expand the wave operator over trees,
and proved that this expansion reduced the number of terms of the 
expansion with respect to usual (tractable) ones, simplified
the denominators of the expansion into a product
of the difference of two energies and lead to powerful formulas 
and recursive computational methods. 

We then showed that this simplification is closely related to the 
algebraic structure of the linear span of permutations
and of a certain convolution subalgebra of trees.

As far as the many-body problem is concerned, we showed
that the simplification of diagrams is not due to the details
of the Hamiltonian but to the general structure of the
wave operator. When the eigenstates and Hamiltonian 
are expressed in terms of creation and annihilation operators
and quantum fields, the algebra of trees mixes with the Hopf 
algebraic structure of fields~\cite{BrouderQG,BrouderMN}. 
It would be interesting to investigate the interplay of
these algebraic structures.

The terms of the Rayleigh-Schr{\"o}dinger series 
are usually considered to be ``quite complicated''
(see ref.~\onlinecite{ReedSimonIV}, p.~8) and difficult
to work with. 
The general term of the 
Rayleigh-Schr\"odinger series for quasi-degenerate systems
is obtained as
the limit for $\epsilon\to0$ of $\Omega_T(0)$ in theorem~\ref{thmRS}.
Through our recursive and non-recursive expressions for these
terms, many proofs of their properties become almost trivial.
The tree structure suggests various resummations of this series,
that will be explored in a forthcoming publication.

\section*{Acknowledgments}
This work was partly supported by the ANR HOPFCOMBOP.
One of the authors (\^A.M.) was supported through the fellowship 
SFRH/BPD/48223/2008 provided by  the Portuguese Science and Technology
Foundation (FCT).

\appendix

\section{A crucial lemma}
\label{crucial}
In this appendix we state and prove a lemma that is crucial
to demonstrate theorem \ref{thmOmegaT}, which is one of the
main results of our paper.
We first need a noncommutative analogue of Chen's formulas for products 
of iterated integrals. This analogue, proved in
refs.~\onlinecite{Agrachev} and \onlinecite{BP-09}, provides a systematic 
link between the theory of iterated integrals, the combinatorics of descents, 
free Lie algebras and noncommutative symmetric functions 
(see ref.~\onlinecite{BP-09} and Appendix~\ref{sect-algebra}
 of the present article for further details). 

Let $L=(L_1,...,L_n)$ be an arbitrary sequence of time-dependent operators
$L_i(t)$, satisfying the same regularity conditions as $H(t)$
in section \ref{tdH}.
Let $\sigma$ be a permutation in $\calS_n$ and define:
$$\Omega^L_\sigma(t,t_0)\defeq\int_{t_0}^{t} \dd t_1 
   \int_{t_0}^{t_1} \dd t_2
    \dots 
   \int_{t_0}^{t_{n-1}} \dd t_n
   L_1(t_{\sigma(1)})\dots L_n(t_{\sigma(n)}).$$

The notation is extended linearly to combinations of permutations, 
so that for 
$\mu\defeq\sum\limits_{\sigma\in \calS_n}\mu_\sigma\sigma$, with
$\mu_\sigma\in\mathbb{C}$. Then
$\Omega^L_\mu(t,t_0)$ is defined as the linear combination
$\sum\limits_{\sigma\in \calS_n} \mu_\sigma\Omega^L_\sigma(t,t_0).$
For $K\defeq (K_1,...,K_m)$ another sequence of time-dependent operators, 
we write $L\cdot K$ for the concatenation product
$(L_1,...,L_n,K_1,...,K_m)$.

We also need to define the convolution product of two
permutations. If $\alpha\in \calS_n$ and $\beta\in\calS_m$,
then $\alpha\ast\beta$ is the sum of the $\binom{n+m}{n}$
permutations $\gamma\in \calS_{n+m}$ such that
$st(\gamma(1),...,\gamma(n))=(\alpha(1),...,\alpha(n))$ and
$st(\gamma(n+1),...,\gamma(n+m))=(\beta(1),...,\beta(m))$. 
Here, $st$ is the standardization map defined 
in section \ref{perm-tree}.
For instance,
\begin{eqnarray*}
(2,3,1)\ast (1) &=& (2,3,1,4)+(2,4,1,3)+(3,4,1,2)+(3,4,2,1),\\
(1,2)\ast (2,1) &=& (1,2,4,3)+(1,3,4,2)+(1,4,3,2)+(2,3,4,1)+(2,4,3,1)
+(3,4,2,1).
\end{eqnarray*}
In words, the product of two permutations $\alpha\in\calS_n$ and
$\beta\in\calS_m$ is the sum of all permutations $\sigma$ of $\calS_{n+m}$
such that the elements of the sequence 
$(\sigma(1),\dots,\sigma(n))$ are ordered as the elements of
$(\alpha(1),\dots,\alpha(n))$, in the sense that 
$\alpha(i)>\alpha(j)$ if and only if $\sigma(i)>\sigma(j)$
and the elements of
$(\sigma(n+1),\dots,\sigma(n+m))$ are ordered as the elements of
$(\beta(1),\dots,\beta(m))$.
   
Now, we can state the noncommutative Chen formula~\cite{Agrachev}
(see also remark 3.3, p.~4111 of ref.~\onlinecite{BP-09})
\begin{lem}
\label{NCChen}
We have:
$$\Omega^L_\alpha(t,t_0)\Omega^K_\beta(t,t_0)=
\Omega^{L\cdot K}_{\alpha\ast\beta}(t,t_0).$$
\end{lem}
The following lemma can be proven similarly:
\begin{lem}\label{fundam}
We have, for $L$ and $K$ as above and $J$ a time-dependent operator:
$$\int_{t_0}^{t} \dd s\Omega^L_\alpha(s,t_0)J(s)\Omega^K_\beta(s,t_0)=
\sum_\gamma\Omega^{L\cdot(J)\cdot K}_\gamma(t,t_0),$$
where $\gamma$ runs over the permutations in $\calS_{n+m+1}$ with 
$\gamma(n+1)=1$, $st(\gamma(1),...,\gamma(n))=\alpha$, 
$st(\gamma(n+2),...,\gamma(n+m+1))=\beta$.
\end{lem}
\begin{proof}
We expand $\Omega^L_\alpha$ and $\Omega^K_\beta$
\begin{eqnarray*}
\int_{t_0}^{t} \dd s\Omega^L_\alpha(s,t_0)J(s)\Omega^K_\beta(s,t_0)
 &=&
\int_{t_0}^{t} \dd s
\int_{t_0}^{s} \dd u_1 \dots \int_{t_0}^{u_{n-1}} \dd u_n
\int_{t_0}^{s} \dd v_1 \dots \int_{t_0}^{v_{m-1}} \dd v_m
\\&&
L_1(u_{\alpha(1)})\dots L_n(u_{\alpha(n)})
J(s)
K_1(v_{\beta(1)})\dots K_m(v_{\beta(m)}).
\end{eqnarray*}
This is the same formula as for the expansion of
$\Omega^L_\alpha(s,t_0)\Omega^K_\beta(s,t_0)$, up to the term $J(s)$ 
and the integration $\int_{t_0}^{t} \dd s$ that however do not change 
the underlying combinatorics. Therefore, lemma~\ref{NCChen} holds in the form
\begin{eqnarray*}
\int_{t_0}^{t} \dd s\Omega^L_\alpha(s,t_0)J(s)\Omega^K_\beta(s,t_0)
 &=&
\sum_{\sigma}
\int_{t_0}^{t} \dd s
\int_{t_0}^{s} \dd s_1 \dots \int_{t_0}^{s_{n+m-1}} \dd s_{n+m}
\\&&
L_1(s_{\sigma(1)})\dots L_n(s_{\sigma(n)})
J(s)
K_1(s_{\sigma(n+1)})\dots K_m(s_{\sigma(n+m)}),
\end{eqnarray*}
where, by the definition of $\alpha\ast\beta$,
the sum over $\sigma$ is over all the permutations
of $\calS_{n+m}$ such that
$st(\sigma(1),\dots,\sigma(n))=\alpha$ and
$st(\sigma(n+1),\dots,\sigma(n+m))=\beta$.
Now, we change variables to
$t_1=s$, $t_{i+1}=s_i$ for $i=1,\dots,n+m$.
The permutation $\gamma$ of $t_1,\dots,t_{n+m+1}$
corresponding to the permutation $\sigma$ of
$s_1,\dots,s_{n+m}$ is characterized by
$\gamma(n+1)=1$ (because $s\ge s_i$ for all $i$),
$\gamma(i)=\sigma(i)+1$ for $1\le i \le n$ and
$\gamma(i+1)=\sigma(i)+1$ for $n+1\le i \le n+m$.
Therefore,
\begin{eqnarray*}
\int_{t_0}^{t} \dd s\Omega^L_\alpha(s,t_0)J(s)\Omega^K_\beta(s,t_0)
 &=&
\sum_{\gamma}
\int_{t_0}^{t} \dd t_1 \dots \int_{t_0}^{t_{n+m}} \dd t_{n+m+1}
\\&&
L_1(t_{\gamma(1)})\dots L_n(t_{\gamma(n)})
J(t_{\gamma(n+1)})
K_1(t_{\gamma(n+2)})\dots K_m(t_{\gamma(n+m+1)}),
\end{eqnarray*}
where $\gamma$ satisfies
$\gamma(n+1)=1$, $st(\gamma(1),...,\gamma(n))=\alpha$, 
$st(\gamma(n+2),...,\gamma(n+m+1))=\beta$.
The lemma is proved.
\end{proof}
Notice that the same process would allow to derive combinatorial formulas 
for arbitrary products of iterated integrals and for integrals with 
integrands involving iterated integrals, provided these products 
and integrands have expressions similar to the ones considered in the 
two lemmas.

\section{The algebraic structure of tree-shaped iterated integrals}
\label{sect-algebra}
Our exposition of tree-parametrized time-dependent perturbation theory has 
focussed on the derivation of convergence results and explicit formulas 
for the time-dependent wave operator. 
However, the reasons why such an approach is possible and efficient are 
grounded into various algebraic and combinatorial properties of trees, 
descents and similar objects.

These properties suggest that the theory of effective Hamiltonians is grounded 
into a new ``Lie theory'' generalizing the usual Lie theory (or, more 
precisely, generalizing the part of the classical Lie theory that is 
relevant to the study of the solutions of differential equations). 
First indications that such a theory exists were already pointed out in 
our ref.~\onlinecite{BP-09}. Indeed, we showed in this article that 
descent algebras 
of hyperoctahedral groups and generalizations thereof are relevant to the 
time-dependent perturbation theory. Applications included  an extension  
of the Magnus expansion for the time-dependent wave operator. Our results 
below provide complementary insights on the subject and further evidence 
that algebraic structures underly many-body theories.

\subsection{The algebra structure}
We know from ref.~\onlinecite{BP-09} and Sect~\ref{permutation-sect} 
that the family 
of integrals $\Omega_\sigma^L(t,t_0)$ is closed under the product. 
This closure property is reflected into the convolution product of 
permutations. This result is a natural noncommutative generalization of 
Chen's formula for the product of iterated integrals.
We show, in the present section, that the same result holds for integrals 
parametrized by trees. We will explain, in the next sections, why such a 
result -which may seem surprising from the analytical point of view- could 
be expected from the modern theory of combinatorial Hopf algebras.

For $L=(L_1,...,L_n)$ a family of time-dependent operators (with the usual 
regularity conditions), and $T$ a tree with $n$ internal vertices, we write 
$$\Omega_T^L(t,t_0)\defeq \sum_{\sigma \in S_T}
   \int_{t_0}^{t} \dd t_1 
   \int_{t_0}^{t_1} \dd t_2
    \dots 
   \int_{t_0}^{t_{n-1}} \dd t_n
   L_1(t_{\sigma(1)}) L_2(t_{\sigma(2)})\dots L_n(t_{\sigma(n)}).$$
This notation is extended to linear combinations of trees, so that e.g. 
if $Z=T+2T'$, where $T$ and $T'$ are two arbitrary trees with the same 
number of vertices, then $\Omega_Z^L(t,t_0)=\Omega_T^L(t,t_0)+
2\Omega_{T'}^L(t,t_0)$.
For $i\le n$, we write $L_{\leq i}=(L_1,...,L_i)$, 
$L_{\geq i}=(L_{i},...,L_n)$.
   
\begin{prop}\label{produ}
For $L=(L_1,...,L_n)$ and $K=(K_1,...,K_m)$ two families of time-dependent 
operators and $T=T_1\vee T_2$, $U=U_1\vee U_2$ two trees, $|T|=n,|T_1|=p,
|T_2|=q,|U|=m, |U_1|=l,|U_2|=k$, we have:
\begin{eqnarray*}
\Omega_T^L(t,t_0)\Omega_U^K(t,t_0) = \int\limits_{t_0}^tds\Omega_T^L(s,t_0)
\Omega_{U_1}^{K_{\leq l}}(s,t_0)K_{l+1}(s)\Omega_{U_2}^{K_{\geq l+2}}(s,t_0)+
\int\limits_{t_0}^tds\Omega_{T_1}^{L_{\leq p}}(s,t_0)L_{p+1}(s)
\Omega_{T_2}^{L_{\geq p+2}}(s,t_0)\Omega_U^K(s,t_0).
\end{eqnarray*}
In the formula, one or several of the trees $T_1,T_2,U_1,U_2$ may be the 
trivial tree $|$.
\end{prop} 
\begin{proof}
Recall the integration by parts formula.
For any integrable functions $f$ and $g$ we define
$F(t)\defeq \int_{t_0}^t \dd s f(s)$,
$G(t)\defeq \int_{t_0}^t \dd s g(s)$ and
$H(t)\defeq F(t) G(t)$.
Then
\begin{eqnarray*}
H(t) &=& \int_{t_0}^t \dd s \frac{\dd H(s)}{\dd s} 
= \int_{t_0}^t \dd s f(s) G(s) + \int_{t_0}^t \dd s F(s) g(s).
\end{eqnarray*}
Now, we use this identity with
$F(t)= \Omega_T^L(t,t_0)$. It follows from the proof of theorem~\ref{OmegaT} 
that $f(s)$ is given by
\begin{eqnarray*}
f(s) =\Omega_{T_1}^{L_{\leq p}}(s,t_0)
L_{p+1}(s)\Omega_{T_2}^{L_{\geq p+2}}(s,t_0),
\end{eqnarray*}
with a similar identity for 
$G(t)=\Omega_U^K(t,t_0)$. The proposition follows.
\end{proof}

In particular, it is a consequence of the proposition and a straightforward 
recursion argument that the linear span of the integrals 
$\Omega_T^L(t,t_0)$ is closed under products.

\subsection{Hopf algebras and Lie theory}

This result may be formalized algebraically. Let us write $\cal T$ for 
the set of formal power series with complex coefficients over the 
set of trees. Proposition~\ref{produ} enables us to define
a product on trees, denoted by $\ast$, such that
$\Omega_T^L(t,t_0)\Omega_U^K(t,t_0)=\Omega_{T\ast U}^{L\cdot K}(t,t_0)$.
This product is defined recursively by the equation~\cite{LodayRonco}
$$T\ast U \defeq (T\ast U_1)\vee U_2+T_1\vee (T_2\ast U).$$

\begin{proof}
The empty tree $\|$ is the unit for the product $\ast$. Assume that
$T\ast U$ is defined and satisfies the recursive relation
for all trees such that $|T|+ |U|< n$, and consider 
two trees $T$ and $U$ with
$|T|+ |U|= n$. The first term on the right hand side of
proposition~\ref{produ} is
$\int_{t_0}^tds\Omega_T^L(s,t_0)
\Omega_{U_1}^{K_{\leq l}}(s,t_0)K_{l+1}(s)\Omega_{U_2}^{K_{\geq
l+2}}(s,t_0)$.
By the recursive relation, we have
\begin{eqnarray*}
\Omega_T^L(s,t_0)
\Omega_{U_1}^{K_{\leq l}}(s,t_0) &=&
\Omega_{T\ast U_1}^{L\cdot K_{\leq l}}.
\end{eqnarray*}
Thus, the whole term can be written
\begin{eqnarray*}
\int_{t_0}^tds\Omega_T^L(s,t_0)
\Omega_{U_1}^{K_{\leq l}}(s,t_0)K_{l+1}(s)\Omega_{U_2}^{K_{\geq
l+2}}(s,t_0) &=& 
\Omega_{(T\ast U_1)\vee U_2}^{L\cdot K}.
\end{eqnarray*}
The second term of proposition~\ref{produ} is
treated similarly and we obtain
\begin{eqnarray*}
\Omega_T^L(t,t_0)\Omega_U^K(t,t_0) &=&
\Omega_{(T\ast U_1)\vee U_2}^{L\cdot K}
+
\Omega_{(T_1\vee(T_2\ast U)}^{L\cdot K}
\end{eqnarray*}
Therefore, the relation
$\Omega_T^L(t,t_0)\Omega_U^K(t,t_0)=\Omega_{T\ast U}^{L\cdot K}(t,t_0)$
gives us
$$T\ast U=(T\ast U_1)\vee U_2+T_1\vee (T_2\ast U).$$
\end{proof}

\begin{cor}
The product provides $\cal T$ with the structure of an 
associative algebra.
\end{cor}
The corollary is a by-product of the associativity of the product 
of operators and of proposition~\ref{produ}. 

Of course, although the analysis of tree-shaped iterated integrals leads 
to a straightforward proof, the associativity property is a purely 
combinatorial phenomenon that originates ultimately from
the associativity of the shuffle product.
See ref.~\onlinecite{BP-09} for the connections between the noncommutative Chen 
formula and shuffle products, see also section 4 of ref.~\onlinecite{sch58} 
for Sch\"utzenberger's classical (but rarely quoted) analysis of the formal 
properties of the shuffle product -in fact, the splitting of the convolution 
product of trees reflects the classical splitting of shuffle products into 
left and right half-shuffle products that had appeared in the study of Lie 
polynomials and was first encoded combinatorially in ref.~\onlinecite{sch58}.
The associativity can also be checked directly or deduced 
from the associativity of the convolution product $\ast$ of permutations, 
since one may verify easily, using e.g. our description of the permutations 
in $S_T$ that the convolution product as introduced above is nothing but 
the restriction to $\cal T$ of the convolution product on the algebra 
${\cal S}=\prod\limits_{n\in\mathbf{N}}{\bf C} [\calS_n]$. 

This fact that the linear span of trees (often written $\bf PBT$) defines a
subalgebra of $\cal S$ for the convolution product (and even a Hopf
subalgebra, referred to as the Hopf algebra of planar binary trees, whereas
$\cal S$ is referred to as the Malvenuto-Reutenauer or Hopf algebra of free
quasi-symmetric functions in the litterature) is well-known. It was first
observed in ref.~\onlinecite{LodayRonco} and further investigated in a 
series of papers~\cite{Chapoton-00-1,Chapoton-00-2,LodayRonco2,Thibon-05}. 
There is however 
a slight subtelty here, since the embedding of $\bf PBT$ in $\cal S$ 
considered e.g. in ref.~\onlinecite{LodayRonco} is not the one we consider 
(another projection map from permutations to trees is used) so that 
$\bf PBT$ and $\cal T$, although isomorphic as algebras (the product rule 
in $\bf PBT$ is the same as the one in $\cal T$, see 
e.g. proposition~3.2 of ref.~\onlinecite{LodayRonco}, where $\bf PBT$ 
is written ${\bf C}[Y_\infty]$) do not agree as subalgebras of $\cal S$. 

However, a corollary of this isomorphism is that the structures existing 
on $\bf PBT$ carry over to the analysis of the algebraic properties of 
tree-shaped iterated integrals. The existence of a Hopf algebra structure 
seems particularly meaningful since the classical applications of the 
theory of free Lie algebras to the analysis of differential equations 
can be rewritten using the framework of Hopf algebras (see e.g. the 
accounts in refs.~\onlinecite{Reutenauer,Patras}).

The same observation holds for the direct sum of the hyperoctahedral group
algebras, that also carries naturally a Hopf algebra structure
\cite{Aguiar-04,Aguiar-08}, and which applications 
to time-dependent perturbative Hamiltonians were studied 
in ref.~\onlinecite{BP-09}.
We leave for further research the investigation of the possible 
role of these Hopf algebra structures in the context of many-body
theories.

\subsection{Permutations, trees and descents}

Another meaningful observation, along the same lines, is that the three
expansions we derived, based respectively on sequences of $P$ and $Q$ (as in
the first expression for $\Omega(t,t_0)$), on trees (third expression) and on
permutations (second expression), reflect at the analytical level the existence
of projection maps and embeddings between hypercubes (or the descent algebra),
planar binary trees,  and permutations. Following Viennot~\cite{Viennot}, these
maps have been at the origin of  modern enumerative combinatorics. We refer to
refs.~\onlinecite{LodayRonco,LodayRonco2,Chapoton-00-1,Chapoton-00-2} 
for a detailed study of these maps that emphasizes the existence of 
underlying geometrical structures that go beyond the Hopf algebraic ones.


\begin{thebibliography}{10}
\bibitem{quadratic}
When $H_0$ is quadratic, the time-dependence of the annihilation operators is
  simply $a_n \ee^{-\imath\epsilon_nt}$, where $\epsilon_n$ is the energy of
  the $n$th eigenvalue of $H_0$ (see, e.g. ref.~\onlinecite{Gross}). When $H_0$
  is not quadratic, the time-dependent annihilation operators and the free
  Feynman propagator cannot be calculated explicitly.
\bibitem{relativelybounded}
This assumption can be relaxed by using the interaction picture (see
  ref.~\onlinecite{ReedSimonII}, p.~283 and \onlinecite{Kato}, p.~557). In the
  adiabatic limit, we recover the Rayleigh-Schr\"odinger series, which is valid
  when the perturbation $V$ is relatively-bounded with respect to $H_0$ (see
  Chapter XII of ref.~\onlinecite{ReedSimonIV}).
\end{thebibliography}
%
 

\end{document}